\documentclass[11pt]{article}
\pdfoutput=1
\usepackage{thumbpdf, amssymb, amsmath, amsthm, microtype, array,
graphicx, verbatim, listings, color, fancybox}
\usepackage[pdftex]{hyperref}
\usepackage{mathtools}
\usepackage{graphicx}
\usepackage{caption}
\usepackage{setspace}
\usepackage{hyperref}
\usepackage{algorithm}
\usepackage{boxedminipage}
\usepackage{mathtools}
\usepackage{algorithmic}
\usepackage{subcaption}
\usepackage{hyperref}
\usepackage{graphicx}
\usepackage{caption}
\usepackage{bbm}
\usepackage[labelformat=simple]{subcaption}
\usepackage{thmtools,thm-restate}

\usepackage[margin=1.2in]{geometry}
\usepackage{multicol,multirow}
\usepackage{tikz}

\newcommand{\Ex}{\mathop{\mathbb{E}}}

\newcommand{\arash}[1]{\noindent{\bf {\color{blue!60!black}{\sc Arash:}  #1}}}

\newcommand{\alantha}[1]{\noindent{\bf {\color{magenta!80!black}{\sc Alantha:}  #1}}}

\allowdisplaybreaks

\declaretheoremstyle[style=claim,qed=$\Diamond$]{claim}
\declaretheoremstyle[style=plain,qed=$\square$]{theorem}

\theoremstyle{plain}
\usetikzlibrary{positioning}

\newtheorem{theorem}{Theorem}
\newtheorem{definition}{Definition}
\newtheorem{corollary}{Corollary}[theorem]

\newtheorem{lemma}{Lemma}
\newtheorem{conjecture}{Conjecture}
\newtheorem{fact}{Fact}
\newtheorem{observation}{Observation}

\theoremstyle{definition}
\newtheorem{define}{Definition}
\theoremstyle{claim}
\newtheorem{claim}{Claim}
\theoremstyle{remark}

\theoremstyle{remark}

\theoremstyle{remark}

\newcommand{\subt}{{\sc Subtour}}

\newcommand{\xcon}{x'}
\newcommand{\xcut}{\bar{x}}

\newenvironment{cproof}
{\begin{proof}
 [Proof.]
 \vspace{-1.5\parsep}
}
{ \end{proof}}

\begin{document}{\bibliographystyle{alpha}}

\title{Shorter tours and longer detours: Uniform covers\\ and a bit beyond}

\author{\textsc{Arash Haddadan\thanks{Tepper School of Business,
      Carnegie Mellon University.  {\tt ahaddada@cmu.edu}},
Alantha Newman\thanks{CNRS and Universit\'e Grenoble-Alpes.  {\tt
      alantha.newman@grenoble-inp.fr}},
R. Ravi\thanks{Tepper School of Business, Carnegie Mellon University. {\tt ravi@cmu.edu}}}}

 \maketitle

\begin{abstract}
Motivated by the well known ``four-thirds conjecture'' for the
traveling salesman problem (TSP), we study the problem of {\em uniform
covers}.  A graph $G=(V,E)$ has an $\alpha$-uniform cover for TSP
(2EC, respectively) if the everywhere $\alpha$ vector
(i.e., $\{\alpha\}^{E}$) dominates a convex combination of incidence
vectors of tours (2-edge-connected spanning multigraphs,
respectively).  The polyhedral analysis of Christofides' algorithm
directly implies that a 3-edge-connected, cubic graph has a 1-uniform
cover for TSP.  Seb\H{o} asked if such graphs have
$(1-\epsilon)$-uniform covers for TSP for some $\epsilon >
0$.  Indeed, the four-thirds conjecture
implies that such graphs have $\frac{8}{9}$-uniform covers.  We show
that these graphs have $\frac{18}{19}$-uniform covers for TSP.  We
also study uniform covers for 2EC and show that the everywhere
$\frac{15}{17}$ vector can be efficiently written as a convex
combination of 2-edge-connected spanning multigraphs.

For a weighted, 3-edge-connected, cubic graph, our results show that
if the everywhere $\frac{2}{3}$ vector is an optimal solution for the
subtour elimination linear programming relaxation for TSP, then a tour
with weight at most $\frac{27}{19}$ times that of an optimal tour can
be found efficiently.  {\em Node-weighted}, 3-edge-connected, cubic
graphs fall into this category.  In this special case, we can apply
our tools to obtain an even better approximation guarantee.

An essential ingredient in our proofs is decompositions of graphs
(e.g., cycle covers) that cover small-cardinality cuts an even
(nonzero) number of times.  Another essential tool we use is
half-integral tree augmentation, which is known to have a small
integrality gap.  To extend our approach to input graphs that are
2-edge-connected, we present a procedure to decompose a
point in the subtour elimination polytope
 into spanning, connected
subgraphs that cover each 2-edge cut an even number of times.  Using
this decomposition, we obtain a
$\frac{17}{12}$-approximation algorithm for minimum weight
2-edge-connected spanning subgraphs on subcubic, node-weighted graphs.

\end{abstract}

\clearpage
\section{Introduction}

The {\bf traveling salesman problem (TSP)} and the minimum-weight {\bf
  2-edge-connected spanning multigraph problem (2EC)} are two
fundamental and well-studied problems in combinatorial optimization.
A folklore conjecture sometimes tersely called the ``four-thirds
conjecture" (see, e.g.,~\cite{CV00,Goemans95}) states that the optimal
(integral) solution for the metric TSP is no more than $\frac43$ times
the value of the subtour elimination linear programming relaxation.
However, the best known approximation ratio for both TSP and 2EC
currently stands at
$\frac32$~\cite{chris,wolsey1980heuristic,shmoys1990analyzing,
  FJ82}. A recent spate of work has focused on the special case of
{\em graph-TSP} when the underlying weights arise from hop distances
in an undirected
graph~\cite{Gharan11,BSSS11,momke2016removing,Mucha14}.  The current
best ratio for this problem is $\frac75$~\cite{sebo12}.  A parallel
line of work has improved the ratio for 2EC in the unweighted case
(commonly referred to as the 2-edge-connected spanning subgraph
problem or {\em 2ECSS} for short) and resulted in a
$\frac{4}{3}$-approximation for this
problem~\cite{sebo12,boyd2014frac}.  So far, these new techniques have
not been extended to more general metrics.

One approach to general metrics is via {\em convex combinations} of
incidence vectors of tours that can be derived from solutions to the
well-known {\em subtour elimination linear programming relaxation},
which we will refer to as \ref{subtour-equal}.  It is by now quite
standard, but we invite the unfamiliar reader to visit Section
\ref{subsec:subtour} for the formal definition.  For a solution
$x\in\;$\ref{subtour-equal}, we use $G_x=(V,E_x)$ to denote the graph
$G=(V,E)$ with edge set $E_x \subseteq E$ restricted to the support of
$x$.  Goemans and Carr and Vempala showed that the four-thirds
conjecture is equivalent to the following
conjecture~\cite{Goemans95,CV00}.
\begin{conjecture}\label{4/3conj-int}
If $x\in$ \ref{subtour-equal}, the vector $\frac{4}{3}x$ dominates a
convex combination of tours in $G_x$.
\end{conjecture}
Based on a polyhedral analysis of Christofides' algorithm, we
know that $\frac{3}{2}x$ dominates a convex combination of tours in
$G_x$~\cite{wolsey1980heuristic,shmoys1990analyzing}; so far we
cannot replace $\frac{3}{2}$ with any smaller constant.  Following the
terminology of Boyd and Seb{\H{o}}~\cite{boydsebo}, for a graph
$G=(V,E)$ on $n$ vertices, let the \textit{everywhere $r$ vector for
  $G$}, be the vector in $\mathbb{R}^{V\choose 2}$ that is $r$ in all
coordinates corresponding to edges of $G$ and 0 in all other
coordinates.  Conjecture \ref{4/3conj-int} is closely related to the
problem of {\em uniform covers}, which we now formally define.
\begin{definition}
A graph $G$ has an $\alpha$-uniform cover for TSP (2EC) if the
everywhere $\alpha$ vector for $G$ dominates a convex combination of incidence
vectors of tours (2-edge-connected spanning multigraphs).
\end{definition}
This close connection is described in Proposition \ref{k-ec,k-reg}.
Proposition \ref{k-ec,k-reg} was observed by Carr and Vempala \cite{CV00} but for completeness we provide a (quite straightforward) proof in Section
\ref{sec:uniform-cover}.

\begin{restatable}{proposition}{introlemma}
\label{k-ec,k-reg}
	The following statements are equivalent.
	
	\begin{itemize}
		\item[(a)] If $x\in$ \ref{subtour-equal}, the vector
                  $\frac{4}{3}x$ dominates a convex combination of tours in $G_x$.
		\item[(b)] For any positive integer $k$ and
                  an arbitrary $k$-edge-connected $k$-regular graph $G$, the everywhere $\frac{8}{3k}$ vector for $G$ dominates a convex combination of tours in $G$.
		\end{itemize}
\end{restatable}

The first interesting case is when $k=3$ (i.e., the case of
3-edge-connected, cubic graphs).  Since the everywhere $\frac{2}{3}$
vector for a 3-edge-connected, cubic graph $G$ is in
\ref{subtour-equal}, Seb{\H{o}} pointed out that the four-thirds
conjecture implies that for a 3-edge-connected, cubic graph, the
everywhere $\frac{8}{9}$ vector dominates a convex combination of
tours~\cite{sebo2015problems}.  The following is therefore a relaxed
version of Conjecture \ref{4/3conj-int}~\cite{sebo2015problems,
  boydsebo}.

\begin{conjecture}\label{8/9conj-int}
Let $G=(V,E)$ be a 3-edge-connected, cubic graph. The everywhere
$\frac{8}{9}$ vector for $G$ dominates a convex combination of tours
of $G$.
\end{conjecture}

For such graphs, the everywhere 1 vector does indeed dominate a convex
combination of tours, which can be shown via the aforementioned
polyhedral proof of Christofides' algorithm by
Wolsey~\cite{wolsey1980heuristic, shmoys1990analyzing}.  In other
words, a 3-edge-connected, cubic graph has a 1-uniform cover.
Seb{\H{o}}~\cite{sebo2015problems} asked if this bound can be improved: Does a
3-edge-connected, cubic graph have a $(1-\epsilon)$-uniform cover (for some small
constant $\epsilon$)?  For the special class of 3-edge-connected,
cubic graphs that are also Hamiltonian, Boyd and Seb{\H{o}} show that
the everywhere $\frac{6}{7}$ vector for $G$ dominates a convex
combination of tours ~\cite{boydsebo}.  We give an affirmative answer
to Seb{\H{o}}'s question and improve this factor from $1$ to
$\frac{18}{19}$ for {\em all} 3-edge-connected, cubic
graphs\footnote{Applying Theorem 2.3 from \cite{philip6/7}, we note
  that this theorem applies to all 3-edge-connected (i.e., possibly
  noncubic) graphs.}.

\begin{restatable}{theorem}{introEighteenNineteen}
\label{18/19}
Let $G=(V,E)$ be a 3-edge-connected, cubic graph. The everywhere
$\frac{18}{19}$ vector for $G$ dominates a convex combination of tours
of $G$ and this convex combination can be found in polynomial time.
\end{restatable}

The same question can be posed replacing tours
with 2-edge-connected spanning multigraphs:
for an arbitrary positive integer $k$ and an arbitrary
$k$-edge-connected $k$-regular graph $G$, can the everywhere
$\alpha_k$ vector be decomposed into a convex combination of
2-edge-connected spanning multigraphs?
For general $k$, the best-known factor for this question
(as in the case for tours)
is $\alpha_k=\frac{3}{k}$, which can be obtained
via the polyhedral proof of Christofides'
algorithm~\cite{wolsey1980heuristic}.  For special cases, however, better factors
are known.  For $k=4$, Carr and Ravi showed that the everywhere
$\frac{2}{3}$ vector can be decomposed into a convex combination of
2-edge-connected spanning multigraphs~\cite{Carr98}.  Their proof is
constructive but is not guaranteed to run in polynomial time.

For a 3-edge-connected, cubic graph (i.e., the case $k=3$), Boyd and
Legault showed that the everywhere $\frac{4}{5}$ vector can be written
as a convex combination of 2-edge-connected spanning
multigraphs~\cite{philip6/7}.  This factor was subsequently improved
to $\frac{7}{9}$ by Legault~\cite{philip}.  These convex combinations
are a key ingredient for a related result on {\em half-triangle
  graphs}~\cite{philip6/7}.  Both the factors $\frac{4}{5}$
and $\frac{7}{9}$ are obtained via constructive procedures that are
not shown to run in polynomial time.
In this paper, we show that for a 3-edge-connected, cubic graph, there
is a polynomial-time algorithm to write the everywhere $\frac{15}{17}$
vector as a convex combination of 2-edge-connected spanning
multigraphs.

\begin{restatable}{theorem}{introFifteenSeventeen}
\label{15/172ec}
Let $G=(V,E)$ be a 3-edge-connected, cubic graph. The everywhere
$\frac{15}{17}$ vector for $G$ dominates a convex combination of
2-edge-connected spanning multigraphs of $G$ and this convex combination can be
found in polynomial time.
\end{restatable}

One implication of Theorem \ref{18/19} is that for a weighted,
3-edge-connected, cubic graph for which the everywhere $\frac{2}{3}$
vector is an optimal solution for \ref{subtour-equal}, we can achieve
an approximation ratio of $\frac{27}{19}$ for TSP, which improves over the
approximation factor of Christofides' algorithm for these
graphs\footnote{We remark that characterizing instances by their
  optimal LP solutions is how classes of fundamental points are
  defined.  Incidentally, many fundamental classes of TSP and 2EC
  extreme points are either cubic or
  subcubic~\cite{Boydcarr,Carr98,CV00}.}.  A natural class of such
graphs are 3-edge-connected, cubic, {\em node-weighted} graphs.  In the
{\em node-weight metric}, each vertex of an undirected graph is
assigned a positive integer weight; the weight of an edge is the sum of the weights of its
two endpoints.
(The node-weight metric is an intermediate class between
weighted and unweighted graphs for studying the TSP and has been previously studied by
Frank~\cite{frank1992augmenting}.)  In fact, we show that using some
of the same tools applied to the uniform cover problems, we can prove
the following improved approximation ratio for such graphs.

\begin{restatable}{theorem}{introSevenFifths}
\label{7/5}
There is a $\frac{7}{5}$-approximation algorithm for TSP on
node-weighted, 3-edge-connected, cubic graphs.
\end{restatable}

Similarly, Theorem \ref{15/172ec} implies that for a weighted,
3-edge-connected, cubic graph for which the everywhere $\frac{2}{3}$
vector is an optimal solution for \ref{subtour-equal}, we can obtain
an approximation ratio of $\frac{45}{34}$ for 2EC, which improves upon
the best-known approximation factor for such graphs derived from
Christofides' algorithm.  We explore this problem further when the
input graph is no longer 3-edge-connected and prove the following
theorem for subcubic, node-weighted graphs.

\begin{restatable}{theorem}{FourThirdsEC}
\label{thm:2EC-node-weight}
If $G$ is a node-weighted, subcubic graph, then there exists a
$\frac{17}{12}$-approximation for 2EC on $G$.
\end{restatable}

\subsection{Outline and Organization}

In Section \ref{sec:notation}, after stating some basic notation, we
formally present the tools we use to prove our main theorems.  The
first tool, presented in Section \ref{subsubsec:BIT-tool}, is an
efficient algorithm by Boyd, Iwata and Takazawa to find a cycle cover
that covers all 3- and 4-edge cuts in a bridgeless, cubic
graph~\cite{bit13}.  This is an essential tool in the proofs of
Theorems \ref{18/19}, \ref{15/172ec} and \ref{7/5}.

In Section \ref{subsubsec:TAP}, we present a key tool for proving
Theorems \ref{15/172ec} and \ref{thm:2EC-node-weight}, which is a
theorem by Cheriyan, Jord\'an and Ravi proving a small integrality gap
for the half-integral 1-cover problem~\cite{cheriyan19992}.  The
1-cover problem generalizes the tree augmentation problem: given a
connected subgraph $S$, the goal is to find an additional subset of
edges (from the edges not in the subgraph $S$) to make $S$
2-edge-connected.  The best-known approximation factor for this
problem is 2~\cite{FJ81}, but when the solution is half-integral,
there is a $\frac{4}{3}$-approximation~\cite{cheriyan19992}.  This
latter result has been generalized by Iglesias and
Ravi~\cite{iglesias2017coloring}.

In Section \ref{sec:uniform-cover}, we show how to apply these tools
to prove our main theorems on uniform covers, which we introduced and
motivated in the introduction.  In Section \ref{sec:BIT}, we show how
to apply these tools to go beyond the approximation guarantee obtained
via uniform covers and present several applications to connectivity
problems on node-weighted, 3-edge-connected, cubic graphs.  In
addition to Theorem \ref{7/5}, we present a
$\frac{13}{10}$-approximation algorithm for 2EC in cubic,
3-edge-connected graphs.  This improves the approximation ratio of
$\frac32$ for this problem that follows from Christofides' algorithm.
A natural question is if we can extend these results to graphs that
are 2-edge-connected and either cubic or subcubic.

Extending our approach to input graphs that are 2-edge-connected
necessitates finding methods for covering 2-edge cuts.  In Section
\ref{sec:decomposition}, we present a procedure to decompose a
solution for the subtour elimination linear program into spanning,
connected sub(multi)graphs that cover each 2-edge cut an even
(nonzero) number of times.  In Section \ref{sec:ala-christofides}, we
demonstrate an application of this decomposition theorem for TSP on
node-weighted, cubic graphs; we show that an algorithm similar to that
of Christofides has an approximation factor better than $\frac{3}{2}$
when the weight of an optimal subtour solution is strictly larger than
twice the sum of the node weights.  In Section \ref{subsec:tap}, we
give another application of our decomposition theorem, which allows
us to (again) apply the aforementioned theorem of Cheriyan, Jord\'an
and Ravi and augment these spanning multigraphs with half-integral
tree augmentations.  Combining this with ideas from Section
\ref{sec:ala-christofides}, we prove Theorem
\ref{thm:2EC-node-weight}.

\section{Preliminaries and Tools}\label{sec:notation}

In the remainder of this paper, $G=(V,E)$ denotes a weighted graph and
$w(e)$ denotes the weight of edge $e \in E$.  We can assume that $G$
is 2-vertex-connected (e.g., applying Lemma 2.1 from
\cite{momke2016removing}).  Graph $G$ is node-weighted if there is a
function $f:V\rightarrow \mathbb{R}^+$ such that for each $e=uv\in E$,
we have $w(e)=f(v)+f(u)$. In this case, we say $G$ is node-weighted
with node-weight function $f$. Denote by $w(E)$ the total edge weight:
$\sum_{e\in E}w(e)$. For ease of notation let $n = |V|$. For vectors
$a,b\in \mathbb{R}^m$ we say $a$ dominates $b$ if $a_i\geq b_i$ in
each coordinate $i \in \{1,\dots, m\}$.

We will work with multisets of edges of $G$. For a multisubset $F$ of
$E$, the submultigraph induced by $F$ (henceforth referred to simply
as a multigraph) is a graph that has the same number of copies of each
edge as in $F$. For a positive integer $t$, the multiset $t\cdot F$ is
the multiset that contains $t$ copies of each element in $F$.  For
multisets $F$ and $F'$, we denote by $F\cup F'$ the multiset that
contains as many copies of each edge as those in $F$ plus those in
$F'$.  For a multiset $F$ and edge $e \in F$, we denote by $F - e$ the
multiset that results from removing a single copy of $e$ from $F$.  By
$F + e$, we denote the multiset that results from adding a single copy
of $e$ to $F$.  For a multiset of edges (or a multigraph) $F$ the
summand $\sum_{e\in F} w(e)$ counts each edge $e\in F$ as many times
as it appears in the multiset $F$.

A multigraph $F$ of $G$ is a tour if the vertex set of $F$ spans $V$,
$F$ is connected, and every vertex in $F$ has even degree.  For the
sake of brevity, we henceforth use the term 2-edge-connected
multigraph of $G$ to refer to a 2-edge-connected {\em spanning}
multigraph (i.e., a multigraph that spans all the vertices of $G$).  For
a subset of edges $S\subseteq E$, the graph $G/S$ is the graph
obtained from $G$ by contracting the edges in $S$ (and deleting
self-loops).  For a subset $S$ of vertices of $G$ let $\delta(S)
\subset E$ denote the edges crossing the cut $(S, V\setminus{S})$.

\subsection{Subtour Elimination Linear Program}\label{subsec:subtour}

Consider a (not necessarily complete) weighted graph $G=(V,E)$ with edge
weights $w(e)$ for $e\in E$.  The output of TSP and 2EC on input
graph $G$ is a minimum weight tour and a minimum weight 2-edge
connected multigraph of $G$, respectively.
The following
relaxation provides a lower bound on the weight of an optimal
solution for both problems.
\begin{align}
&z_G=\min\sum_{e\in E} w(e) x_e & \;\nonumber \\
&x(\delta(S))\geq 2
  &\text{for } \emptyset \subset S \subset V \tag{{\sc Subtour}$(G)$}\label{subtour} \\
& x_e \geq 0 & \text{for } e \in  E. \nonumber
\end{align}
The \textit{metric completion of $G$} is the complete graph $G^{met}$
on the vertex set of $G$ such that for $u,v\in V$ the weight of the
$uv$ edge in $G^{met}$ is the weight of the shortest path between $u$
and $v$ in $G$. Clearly, these weights obey the triangle
inequality. TSP on $G$ is equivalent to finding a minimum weight tour
in $G^{met}$.  Since $G^{met}$ contains a minimum weight tour that is
a Hamilton cycle, the following degree constraints are valid and yield
the following seemingly stronger lower bound for TSP.
\begin{align}
&z_{G^{met}}=\min\sum_{u,v\in V} w(uv) x_{uv} & \;\nonumber \\
&\sum_{u\in S,v\notin S} x_{uv} \geq 2
&\text{for } \emptyset \subset S \subset V
\tag{{\sc
		Subtour}$^{=}(G)$}\label{subtour-equal} \\
& \sum_{v\in V\setminus \{u\}} x_{uv} =  2 & \text{for } u\in V\nonumber\\
& x_{uv} \geq 0 & \text{for } u,v \in  V. \nonumber
\end{align}
Note that in the above formulation, edges $uv$ and $vu$ are the same, so
$x_{uv}$ and $x_{vu}$ represent the same variable. Cunningham showed that the bounds
$z_G$ and $z_{G^{met}}$ are in fact equal~\cite{Monma1990,Goeberts}.
For a
solution $x\in\;$\ref{subtour-equal} let $G_x=(V,E_x)$, where
$E_x=\{uv : u,v \in V \text{ and }
x_{uv}>0\}$.

We will frequently use the following well-known fact~\cite{lgs}.
\begin{fact}\label{fact:spanning-tree}
Any point $x \in $ \ref{subtour} dominates a convex combination of spanning
trees, which can be found efficiently.
\end{fact}

\subsection{Cycle Covers Covering All 3- and 4-Edge Cuts}\label{subsubsec:BIT-tool}

We now present one of our main tools for proving Theorems \ref{18/19} and \ref{15/172ec}.
Given a graph $G=(V,E)$, a \textit{cycle cover} (also known as a
{\em 2-factor}) of $G$ is a collection of vertex disjoint cycles whose
vertex sets partition $V$. Cycle covers have been extensively studied
in the area of matching theory and have been also used to obtain
approximation algorithms for TSP.

Kaiser and \v{S}krekovski \cite{kaiser} proved that every bridgeless,
cubic graph has a cycle cover that covers all 3-edge and 4-edge cuts
of the graph.  Their proof is not algorithmic and an efficient, constructive
version was given by Boyd, Iwata and Takazawa~\cite{bit13}.

\begin{theorem}[Boyd, Iwata and Takazawa \cite{bit13}]\label{bit13}
Let $G$ be a bridgeless, cubic graph. Then there is an algorithm whose
running time is polynomial in the size of $G$ that finds a cycle cover
of $G$ covering every 3-edge and 4-edge cut of $G$.
\end{theorem}

A straightforward observation is the following.
\begin{observation}\label{5ec}
Let $G$ be a 3-edge-connected, cubic graph. Let $C$ be a cycle cover
that covers 3-edge cuts and 4-edge cuts in the graph. Then $G/C$ is a
5-edge-connected multigraph.
\end{observation}

Cubic, bipartite graphs exhibit even more structure, allowing for a
stronger corollary.

\begin{observation}\label{bip-cycle-contract}
	Let $G$ be a cubic, bipartite graph. Let $C$
        be a cycle cover of $G$.  Then the graph $G/C$ is Eulerian.
\end{observation}

\begin{proof}
Each vertex in $G/C$ corresponds to a cycle in $C$ and the degree of
this vertex has the same parity as the number of edges in the cycle.  Since
$G$ is bipartite, every cycle in $C$ is an even cycle.  Therefore,
each vertex in $G/C$ has even degree, since it is obtained by
contracting a cycle in $C$.  We can conclude that $G/C$ is an Eulerian
graph.
\end{proof}

\begin{observation}\label{bipartitebit}
	Let $G$ be a 3-edge-connected, cubic, bipartite graph. Let $C$
        be a cycle cover that covers 3-edge cuts and 4-edge cuts in
        the graph. Then $G/C$ is a 6-edge-connected graph.
\end{observation}
\begin{proof}
	Graph $G/C$ is 5-edge-connected by Observation \ref{5ec}.  By
        Lemma \ref{bip-cycle-contract}, $G/C$ is Eulerian.  Therefore,
        $G/C$ does not contain any cuts crossed by an odd number of
        edges.  In particular, $G/C$ contains no 5-edge cuts.
\end{proof}

\subsection{Tree Augmentation}\label{subsubsec:TAP}

We next present one of our main tools for proving Theorem
\ref{15/172ec}.  We first state the 1-cover problem on a laminar
family of sets. A family of sets $\mathcal{S}$ is called {\em laminar}
if for any $S$ and $S'$ in $\mathcal{S}$, the set $S\cap S'$ is equal
to either $S$, $S'$ or $\emptyset$.  For a graph $G=(V,E)$, we are
given a laminar family of sets, $\mathcal{S}$, where each set in
$\mathcal{S}$ consists of a subset of vertices.  Additionally, we are
given a set of edges $E$ with nonnegative edge weights $w(e)$ for
$e\in E$.  The 1-cover problem on family $\mathcal{S}$ asks for a
\textit{1-cover of $\mathcal{S}$}: a minimum weight subset of edges $C
\subseteq E$ such that $|C\cap \delta(S)|\geq 1$ for all $S \in
\mathcal{S}$. Indeed, we are interested in a special case of the
1-cover problem on a laminar family of sets. Let $F$ be a spanning,
connected multigraph of a given graph $G$, and let $\mathcal{S}$ be
the family of 1-edge cuts of $F$: $\mathcal{S} = \{S:~|\delta(S)\cap
F|=1\}$. In this case, we refer to a 1-cover of $\mathcal{S}$ as a
\textit{1-cover of $F$}. Define a {\em block}
to be a maximal 2-edge-connected induced subgraph of $F$. Consider the tree obtained
from contracting the blocks of $F$. Rooting this tree
at an arbitrary vertex, we can find a laminar family $\mathcal{S}_F$
of sets in $\mathcal{S}$ such that the 1-covers of $\mathcal{S}$ are
exactly the 1-covers of $\mathcal{S}_F$ . Hence, the natural linear
programming relaxation for the 1-cover problem for a graph $G = (V,E)$
and multigraph $F$ of $G$ is:
\begin{align}
\min & \sum_{e\in E} w(e) y_e& \nonumber\\
s.t. & \sum_{e\in E: e\in \delta(S)} y_e \geq 1 & \mbox{for all $S\in \mathcal{S}_F$}\tag{{\sc Cover}$(G,F)$}\label{C`over}\\
& y_e \geq 0 &\mbox{for all $e\in E$}.\nonumber
\end{align}

Let us denote the feasible region of the above linear program by {\sc
  Cover}$(G,F)$.  By contracting the blocks of 
$F$, we get a tree on these contracted components and the 1-cover
problem on $\mathcal{S}_F$ is 
equivalent to the tree augmentation problem~\cite{FJ81}.  Its
integrality gap is known to be between $\frac{3}{2}$ and $2$
\cite{FJ81,Jain2001,32gap}. However, in the special case of
half-integral points, the integrality gap is much smaller. Cheriyan, Jordan and Ravi \cite{cheriyan19992} proved that if $y\in$ {\sc Cover}$(G,F)$ and $y_e\in \{0,\frac{1}{2},1\}$
for all $e\in E$, then there is an algorithm, whose running time
is polynomial in the size of $G$, that writes the vector
$\frac{4}{3}\cdot y$ as a convex combination of 1-covers
$C_1,\ldots,C_h$ of $F$. Iglesias and Ravi generalized this result \cite{iglesias2017coloring}.

\begin{theorem}[Iglesias and Ravi \cite{iglesias2017coloring}]\label{cjr}
	If $y\in$ {\sc Cover}$(G,F)$ and $y_e\geq \alpha$ or $y_e=0$
        for all $e\in E$, then there is an algorithm, whose running time
        is polynomial in the size of $G$, that writes the vector
        $\frac{2}{1+\alpha}\cdot y$ as a convex combination of 1-covers
        $C_1,\ldots,C_h$ of $F$.
\end{theorem}

\section{Uniform Covers}\label{sec:uniform-cover}

First, we recall Proposition \ref{k-ec,k-reg}, stated in the introduction. This observation is due to Carr and Vempala \cite{CV00}, but we prove the proposition for completeness.
\introlemma*

\begin{proof}
	(a)$\implies$(b): If $G$ is $k$-edge-connected and
  $k$-regular, then $y$, defined to be the everywhere $\frac{2}{k}$ vector for $G$, is in
  \ref{subtour-equal}. Therefore $\frac{4}{3}y$, which is the
  everywhere $\frac{8}{3k}$ vector for $G$, is dominated by
  a convex combination of tours of $G_y$. Notice that $G_y=G$.
	
	(b)$\implies$(a): Let $x\in $ \ref{subtour-equal} for
  graph $G=(V,E)$.  Let $k$ be the smallest integer such that $x_e$ is
  a multiple of $\frac{1}{k}$ for every edge $e \in E_x$.  Let
  $G'=(V,E')$ be such that $E'$ has $kx_e$ copies of each $e\in
  E_x$. It is easy to observe that $G'$ is $2k$-edge-connected and
  $2k$-regular.  Let $y$ be the everywhere $\frac{8}{6k}$ vector for
  $G'$.  So by (b), $y$ dominates a convex combination of tours in
  $G'$: $y\geq \sum_{i=1}^{\ell}\lambda_i \chi^{F_i}$, where
  $\sum_{i=1}^{\ell}\lambda_i=1$, $\lambda_i>0$, and $F_i$ is a tour
  of $G'$ for $i=\{1,\ldots,\ell\}$.  Since $G'= k \cdot G_x$, each $F_i$ also
  corresponds to a tour in $G_x$, and $\sum_{i=1}^{\ell}\lambda_i
  \chi^{F_i}(e)=\frac{8}{6k}kx_e = \frac{4}{3}x_e.$
\end{proof}

\subsection{Algorithms for Uniform Covers}

Recall that the polyhedral proof of Christofides' algorithm can be
used to prove statement (b) in Proposition \ref{k-ec,k-reg} when the
factor $\frac{8}{3k}$ is replaced by $\frac{3}{k}$.  The problem of
reducing this factor to anything less than $\frac{3}{k}$ has been open
for decades.  In the case where $k=3$, we can improve this result.

\introEighteenNineteen*

\begin{proof}
By Theorem \ref{bit13}, graph $G$ has a cycle cover $C$ such that $C$
covers every 3-edge and 4-edge cut of $G$. Let $G/C$ be the graph
obtained by contracting each cycle of $C$ in $G$. By Observation
\ref{5ec}, $G/C$ is 5-edge-connected. Define vector $y\in
\mathbb{R}^{E(G / C)}$ as follows: $y_e = \frac{2}{5}$ for $e\in
E(G/C)$. Observe that $y\in$ \subt$(G/C)$. Thus, $y$ dominates a
convex combination of spanning trees of $G/C$, which can be computed
in polynomial time (see Fact \ref{fact:spanning-tree}). More
precisely, we can write $y \geq \sum_{i=1}^{\ell} \lambda_i
\chi^{T_i}$, where $T_i$ is a spanning tree of $G/C$,
$\sum_{i=1}^{\ell} \lambda_i =1$, and $\lambda_i>0$ for $i \in
\{1,\ldots,\ell\}$. Consequently, we have $2 y \geq\sum_{i=1}^{\ell}
\lambda_i \chi^{2 T_i}$ (i.e., the vector $2 y$ dominates a convex
combination of doubled spanning trees of $G/C$).

 Let $M$ be the set of edges in $E\setminus C$ that are not in $G/C$;
 these are the edges that connect two vertices of the same cycle in
 $C$. Define vector $v\in \mathbb{R}^{V\choose 2}$ as follows: $v_e =
 1$ for $e\in C$, $v_e = \frac{4}{5}$ for $e\in E\setminus (M\cup C)$,
 and $v_e=0$ otherwise. Note that $v \geq \sum_{i=1}^{\ell} \lambda_i
 \chi^{C\cup 2 T_i}$. For $i\in\{1,\ldots,\ell\}$, the graph induced
 by $C\cup 2 T_i$ is a tour.
	
Now we define $u\in \mathbb{R}^{V\choose 2}$ as follows: $u_e=
\frac{1}{2}$ for $e\in C$ and $u_e=1$ for $e\in E\setminus C$, and
$u_e=0$ otherwise. We have $u\in$ \ref{subtour-equal} : for each cut
$D$ of $G$, if $|D|\geq 4$, clearly $\sum_{e \in D} u_e \geq 2$. If
$|D|=3$, then $|C\cap D|= 2$, so $\sum_{e \in D}u_e =
2\cdot \frac{1}{2} + 1 = 2$. We can write $\frac{3}{2} u$ as a
convex combination of tours~\cite{wolsey1980heuristic}.
	
Now vector $\frac{15}{19}v + \frac{4}{19}(\frac{3}{2}u)$ can be
written as convex combination of tours of $G$. For edge $e\in C$ we
have $\frac{15}{19}v_e+\frac{4}{19}(\frac{3}{2}u_e) =
\frac{15}{19}+\frac{4}{19}(\frac{3}{2}\cdot
\frac{1}{2})=\frac{18}{19}$.  For $e\in E(G/C)$ we have
$\frac{15}{19}v_e+\frac{4}{19}(\frac{3}{2}u_e) = \frac{15}{19}\cdot
\frac{4}{5}+\frac{4}{19}(\frac{3}{2})=\frac{18}{19}$.  For $e\in M$,
we have $\frac{15}{19}v_e+\frac{4}{19}(\frac{3}{2}u_e) = 0
+\frac{4}{19}(\frac{3}{2})=\frac{6}{19}$. Therefore $\frac{15}{19}v +
\frac{4}{19}(\frac{3}{2}u)$ is dominated by the everywhere
$\frac{18}{19}$ vector for $G$.  \end{proof}

If $G$ is also bipartite, then by Observation \ref{bipartitebit}, the
graph $G/C$ in the proof of Theorem \ref{18/19} is 6-edge connected.
We can therefore improve Theorem \ref{18/19} in this case.

\begin{theorem}\label{12/13}
	Let $G=(V,E)$ be a 3-edge-connected, cubic, bipartite
        graph. The everywhere $\frac{12}{13}$ vector for $G$ dominates
        a convex combination of tours of $G$ and this convex
        combination can be found in polynomial time.
\end{theorem}
\begin{proof}
Let $C$ be the cycle cover in $G$ that covers 3-edge and 4-edge cuts
of $G$. By Observation \ref{bipartitebit}, $G/C$ is
6-edge-connected. Let $M$ be the set of edges that have both endpoints
in the same cycle in the cycle cover $C$.  Similar to the proof of
Theorem \ref{18/19}, define vector $v \in \mathbb{R}^{V\choose 2}$ as
follows: $v_e=1$ for $e\in C$, $v_e=\frac{2}{3}$ for $e\in E(G/C)$,
and $v_e=0$ otherwise. The vector $v$ can be written as a convex
combination of tours of $G$.
	
Now define $u \in \mathbb{R}^{V \choose 2}$ as follows: $u_e=\frac{1}{2}$ for
$e\in C$, $u_e=1$ for $e\in E\setminus C$, and $u_e=0$
otherwise. Since $u\in\;$\ref{subtour-equal}, this implies that
$\frac{3}{2}u$ can be written as a a convex combination of tours of
$G$.
	
Finally, vector $\frac{9}{13}v + \frac{4}{13}(\frac{3}{2}u)$ can be
written as a convex combination of tours of $G$. For $e\in C$,
$\frac{9}{13}v_e + \frac{4}{13}u_e=
\frac{9}{13}+\frac{4}{13}(\frac{3}{4})= \frac{12}{13}$. For $e\in
E(G/C)$ we have $\frac{9}{13}v_e + \frac{4}{13}u_e= \frac{9}{13}\cdot
\frac{2}{3}+\frac{4}{13}(\frac{3}{2})= \frac{12}{13}$. Finally, if
$e\in M$, $\frac{9}{13}v_e + \frac{4}{13}u_e=
\frac{4}{13}(\frac{3}{2})= \frac{6}{13}$. This proves the result.
\end{proof}

We can further relax Conjecture \ref{8/9conj-int} and ask whether or
not the everywhere $\frac{8}{9}$ vector for a 3-edge-connected, cubic
graph $G$ can be written as a convex combination of 2-edge-connected
multigraphs of $G$.  The answer to this question is yes, and is a
direct corollary of a decomposition theorem due to Carr and Ravi
\cite{Carr98}. In fact, in Lemma \ref{2ec8/9}, we show that for a
3-edge-connected, cubic graph $G$, the problem of determining if $G$
has an $\alpha$-uniform cover can be reduced to bounding the
integrality gap of half-integer solutions for \textsc{Subtour}$^=(G)$.

\begin{lemma}\label{2ec8/9}
	Let $G=(V,E)$ a 3-edge-connected, cubic graph. Suppose for any
        $x\in$ \textsc{Subtour}$^=(G)$ such that $x\in
        \{0,\frac{1}{2},1\}^{{V}\choose 2}$, the vector $\alpha x$
        dominates a convex combination of tours (2-edge-connected
        multigraphs) of $G_x$.  Then the everywhere
        $\frac{2}{3}\alpha$ vector for $G$ dominates a convex
        combination of tours (2-edge-connected multigraphs) of $G$.
\end{lemma}	
\begin{proof}
	Let $y\in \mathbb{R}^{E}$ be such that $y_e=\frac{2}{3}$ for $e\in
	E$. By Corollary 30.8a in \cite{schrijver}, $y$ is in the convex hull
	of cycle covers of $G$.  Thus, there are cycle covers
	$C_1,\ldots,C_\ell$ and positive multipliers
	$\lambda_1,\ldots,\lambda_\ell$ such that
	$\sum_{i=1}^{\ell}\lambda_i=1$ and $y= \sum_{i=1}^{\ell}\lambda_i
	C_i$.  For $i \in \{1,\ldots,\ell\}$, define vector $y^i\in
	\mathbb{R}^{V\choose 2}$ as follows: $y^i_e=\frac{1}{2}$ for $e\in
	C_i$, $y^i_e=1$ for $e\in E\setminus C_i$, and $y^i_e=0$
	otherwise. Observe that $y^i\in$ \textsc{Subtour}$^=(G)$ for $i \in
	\{1,\ldots,\ell\}$. Furthermore, $v=\sum_{i=1}^{\ell}\lambda_i y^i$ is
	the everywhere $\frac{2}{3}$ vector for $G$; for $e\in E$ we have
	$\sum_{i:e \in C_i} \lambda_i=\frac{2}{3}$ and $\sum_{i:e \notin
		C_i}\lambda_i=\frac{1}{3}$, and so $\sum_{i=1}^\ell \lambda_i y_e^i
	= \sum_{i: e\in C_i}\lambda_i \cdot \frac{1}{2}+\sum_{i: e\notin
		C_i}\lambda_i = \frac{2}{3}$.
	
	Since the vector $y^i \in \{0,\frac{1}{2},1\}^{{V}\choose 2}$, 
	the vector $\alpha y^i$ dominates a convex combination of tours (2-edge-connected multigraphs) of $G_{y^i}=G$ for $i \in
	\{1,\ldots,\ell\}$. Therefore, the everywhere $\frac{2}{3}\alpha$ vector
	for $G$ 	 dominates a convex combination of tours (2-edge-connected multigraphs) of $G$.
\end{proof}

\begin{theorem}[Carr and Ravi \cite{Carr98}]\label{thm:CR98}
	If $x\in$ \ref{subtour-equal} and $x\in\{0,\frac{1}{2},1\}^{V \choose 2}$, then the
	vector $\frac{4}{3} x$ dominates a convex combination of
	2-edge-connected multigraphs of $G_x$.
\end{theorem}

\begin{corollary}\label{cor8.1}
	Let $G=(V,E)$ be a 3-edge-connected, cubic graph. The everywhere
	$\frac{8}{9}$ vector for $G$ dominates a convex combination of
	2-edge-connected multigraphs of $G$.
\end{corollary}
\begin{proof}
	Follows directly from Lemma \ref{2ec8/9} and Theorem \ref{thm:CR98}.
\end{proof}

However, this proof does not yield a polynomial-time decomposition
since the number of multigraphs in the convex combination output via
Theorem \ref{thm:CR98} is not guaranteed to be polynomial in the size
of $G$.  In fact, Legault proved a result that is stronger than Lemma
\ref{2ec8/9}: the everywhere $\frac{7}{9}$ vector for $G$ can be
written as a convex combination of 2-edge-connected
subgraphs~\cite{philip}.  Notice that the result of Legault is
stronger not only because the $\frac{7}{9}$ is smaller than
$\frac{8}{9}$, but also in the sense that it restricts the multigraphs
to subgraphs, i.e. no edge in $G$ is doubled.  However, the proof in
\cite{philip} also does not guarantee that the number of subgraphs in
the decomposition is polynomial in the size of $G$.

We now present a stronger version of Corollary \ref{cor8.1}. For the
rest of this section we will do all computations on the edges of the
graph $G=(V,E)$ so all the vectors are in dimension of $E$.  Thus,
henceforth we slightly abuse the everywhere vector notation to make
the presentation simpler. Indeed, we can extend all the vectors to
dimension $V\choose 2$ by adding zeros.

\begin{theorem}\label{2ec8/9cons}
Let $G=(V,E)$ be a 3-edge-connected, cubic graph. The everywhere
$\frac{8}{9}$ vector for $G$ dominates a convex combination of
2-edge-connected subgraphs of $G$ and this convex combination can be
found in polynomial time.
\end{theorem}

\begin{proof}
Let $y\in \mathbb{R}^E$ be such that $y_e=\frac{2}{3}$ for $e\in E$.  Since $y\in $
\ref{subtour}, we can find in polynomial time
spanning trees $T_1,\ldots,T_\ell$ of $G$ and positive multipliers
$\lambda_1,\ldots,\lambda_\ell$ such that $\sum_{i=1}^{\ell}\lambda_i
=1$ and $y\geq\sum_{i=1}^{\ell} \lambda_i \chi^{T_i}$.  For $i \in
\{1,\ldots,\ell\}$ define vector $y^i\in \mathbb{R}^E$ as follows:
$y^i_e= 0$ for $e\in T_i$ and $y^i_e=\frac{1}{2}$ for $e\notin T_i$.
Since $G$ is 3-edge-connected, we have $y^i\in$ {\sc Cover}$(G,T_i)$
for $i \in \{1,\ldots,\ell\}$.
By Theorem \ref{cjr}, there is a
polynomial-time algorithm that finds 1-covers
$C^i_1,\ldots,C^i_{\ell_i}$ of $T_i$ for $i \in \{1,\ldots,\ell\}$ and
positive multipliers $\lambda^i_1,\ldots,\lambda^i_{\ell_i}$ such that
$\sum_{j=1}^{\ell_i}\lambda^i_j = 1$ and
$\frac{4}{3} y^i = \sum_{j=1}^{\ell_i} \lambda^i_j \chi^{C^i_j}$
for $i \in \{1,\ldots,\ell\}$. Note that $T_i+C^i_j$ is a 2-edge-connected
subgraph of $G$ for $i\in\{1,\ldots,\ell\}$ and
$j\in\{1,\ldots,\ell_i\}$. Hence,
$$u=\sum_{i\in \{1,\ldots,\ell\}} \sum_{j\in \{1,\ldots,\ell_i\}}
\lambda_i \lambda^i_j \chi^{T_i\cup C^i_j}, \quad \text{where
}\sum_{i\in \{1,\ldots,\ell\}} \sum_{j\in \{1,\ldots,\ell_i\}}
\lambda_i \lambda^i_j =1$$ is a convex combination of
2-edge-connected multigraphs of $G$. By construction, an edge cannot
belong both to a tree $T_i$ and to a 1-cover $C^i_j$.  Thus, there are
no doubled edges in any solution.  Vector $u$ is the everywhere
$\frac{8}{9}$ vector for $G$: for $e\in E$, we have $$u_e = \sum_{i:
  e\in T_i} \sum_{j=1}^{\ell_i}\lambda_i\lambda^i_j + \sum_{i: e\notin
  T_i} \sum_{j: e\in C^i_j}\lambda_i\lambda^i_j \leq \frac{2}{3} +
\frac{1}{3}\cdot \frac{2}{3} = \frac{8}{9}.$$
\end{proof}

Observe that in the proof of Lemma \ref{2ec8/9cons}, we
    never double any edge in any of the 2-edge-connected
    subgraphs.  (Hence, the statement of lemma uses {\em subgraph}
    rather than {\em multigraph}.)
If we relax this and allow doubled edges, we can indeed improve the factor by
    combining the ideas from Theorem \ref{18/19} and Theorem
    \ref{2ec8/9cons} to improve the bound in Theorem \ref{2ec8/9cons}
    from $\frac{8}{9}$ to $\frac{15}{17}$.

\introFifteenSeventeen*
\begin{proof}
Let $C$ be a cycle cover of $G$ that covers every 3-edge and 4-edge
cut of $G$. By Observation \ref{5ec}, the graph $G/C$ is
5-edge-connected. Let $M=E\setminus (C\cup E(G/C))$. Define $r \in
\mathbb{R}^{E(G/C)}$ as follows: $r_e = \frac{2}{5}$ for $e\in E(G/C)$.  We
have $r\in$ \subt$(G/C)$, so $\frac{3}{2} r$ dominates a convex
combination of tours of $G/C$: namely $R_1,\ldots, R_\ell$. Observe
that the graph induced by $C\cup R_i$ is a 2-edge-connected multigraph
of $G$ for $i \in \{1,\ldots,\ell\}$. So, the vector $v\in
\mathbb{R}^E$ where $v_e = 1$ for $e\in C$, $v_e = \frac{3}{5}$ for
$e\in E(G/C)$, and $v_e = 0$ for $e\in M$ dominates a convex
combination of 2-edge-connected multigraphs of $G$.

Now define $y\in \mathbb{R}^{E}$ as follows: $y_e = \frac{1}{2}$ for
$e\in C$ and $y_e=1$ for $e\in E\setminus C$. Since $y\in$ \subt$(G)$,
we can efficiently find spanning trees $T_1,\ldots,T_\ell$ of $G$ and
convex multipliers $\lambda_1,\ldots,\lambda_{\ell}$ such that $y\geq
\sum_{i=1}^{\ell} \lambda_i\chi^{T_i}$. For $i \in \{1,\ldots,\ell\}$ define
$y^i\in \mathbb{R}^E$ as follows:
$y^i_e=\frac{1}{2}$ for $e\notin T_i$ and $y^i_e=0$ otherwise. Notice, that $y^i\in$ {\sc
	Cover}$(G,T_i)$, hence by Theorem \ref{cjr}, there is a
polynomial-time algorithm that finds 1-covers
$C^i_1,\ldots,C^i_{\ell_i}$ of $T_i$ for $i \in \{1,\ldots,\ell\}$ and
positive multipliers $\lambda^i_1,\ldots,\lambda^i_{\ell_i}$ such that
$\sum_{j=1}^{\ell_i}\lambda^i_j = 1$ and
$\frac{4}{3} y^i = \sum_{j=1}^{\ell_i} \lambda^i_j \chi^{C^i_j}$
for $i \in \{1,\ldots,\ell\}$. Note that $T_i+C^i_j$ is a 2-edge-connected
subgraph of $G$ for $i\in\{1,\ldots,\ell\}$ and
$j\in\{1,\ldots,\ell_i\}$. Hence,
$$u=\sum_{i\in \{1,\ldots,\ell\}} \sum_{j\in \{1,\ldots,\ell_i\}}
\lambda_i \lambda^i_j \chi^{T_i\cup C^i_j}, \quad \text{where
}\sum_{i\in \{1,\ldots,\ell\}} \sum_{j\in \{1,\ldots,\ell_i\}}
\lambda_i \lambda^i_j =1$$ is a convex combination of
2-edge-connected multigraphs of $G$. For $e\in C$, we have
$$u_e =
\sum_{i: e\in T_i} \sum_{j=1}^{\ell_i}\lambda_i\lambda^i_j + \sum_{i:
  e\notin T_i} \sum_{j: e\in C^i_j}\lambda_i\lambda^i_j \leq \frac{1}{2}
+ \frac{1}{2}\cdot \frac{2}{3}= \frac{5}{6}.$$
For $e\notin C$, we
have $$u_e = \sum_{i: e\in T_i}
\sum_{j=1}^{\ell_i}\lambda_i\lambda^i_j + \sum_{i: e\notin T_i}
\sum_{j: e\in C^i_j}\lambda_i\lambda^i_j \leq 1+0 =1.$$

Finally we
conclude that the vector $\frac{5}{17} v+\frac{12}{17}u$ can be
efficiently written as convex combination of 2-edge-connected
multigraphs of $G$.
For $e\in
C$ we have $\frac{5}{17} v_e+\frac{12}{17}u_e = \frac{5}{17}+
\frac{12}{17}\cdot \frac{5}{6}= \frac{15}{17}$.  For $e\in G/C$ we have
$\frac{5}{17} v_e+\frac{12}{17}u_e = \frac{5}{17}\cdot \frac{3}{5}+
\frac{12}{17}= \frac{15}{17}$.  For $e\in M$ we have
$\frac{5}{17} v_e+\frac{12}{17}u_e = \frac{5}{17}\cdot 0+ \frac{12}{17}=
\frac{12}{17}$.
Therefore $\frac{5}{17} v+\frac{12}{17}u$ is
dominated by the everywhere $\frac{15}{17}$ vector for $G$.
\end{proof}

We note that in the proof of Theorem \ref{15/172ec}, since the vector
$y$ is half-integral, we can apply Theorem \ref{thm:CR98} to conclude
that $\frac{4}{3}y$ dominates a convex combination of 2-edge-connected
multigraphs of $G$.  This shows that the everywhere $\frac{7}{8}$
vector for $G$ dominates a convex combination of 2-edge-connected
multigraphs.  (Specifically, $\frac{3}{8}(\frac{4}{3}y) +
\frac{5}{8}v$ is dominated by the everywhere $\frac{7}{8}$ vector for
$G$.)  But this approach does not produce a convex combination in
polynomial-time.  We can improve Theorem \ref{15/172ec} slightly when
the graph $G$ is also bipartite.

\begin{theorem}\label{7/8bip}
Let $G=(V,E)$ be a 3-edge-connected, cubic, bipartite graph. The everywhere
$\frac{7}{8}$ vector for $G$ dominates a convex combination of
2-edge-connected multigraphs of $G$ and this convex combination can be
found in polynomial time.	
\end{theorem}
\begin{proof}
Let $C$ be the cycle cover in $G$ that covers 3-edge and 4-edge cuts
of $G$. Let $M$ be the set of edges in $G$ that have both endpoints in
the same cycle of $C$. Since $G/C$ is 6-edge-connected, the vector $r$
with $r_e=\frac{1}{3}$ for $e\in E(G/C)$ is in
\subt$(G/C)$. Therefore, we can show, similarly as in the proof of Theorem \ref{15/172ec},
that the vector $v$ such that $v_e=1$ for $e\in C$ and $v_e =
\frac{3}{2}\cdot \frac{1}{3}=\frac{1}{2}$ for $e\in E(G/C)$ and
$v_e=0$ for $e\in M$ can be written as a convex combination of
2-edge-connected multigraphs of $G$ in polynomial time. Furthermore,
as in the proof of Theorem \ref{15/172ec}, the vector $u$, where
$u_e=\frac{5}{6}$ for $e\in C$, $u_e=1$ for $e\in E \setminus C$, can
be written as a convex combination of 2-edge-connected subgraphs of
$G$ in polynomial time.  Note that the vector
$\frac{1}{4}v+\frac{3}{4}u$ is dominated by the everywhere
$\frac{7}{8}$ vector for $G$.
\end{proof}

For the case where $k=4$ in Proposition \ref{k-ec,k-reg}, Carr and
Ravi \cite{Carr98} showed that the everywhere $\frac{2}{3}$ vector can
be written as a convex combination of 2-edge-connected
subgraphs\footnote{\cite{Carr98} do not double half edges, so in fact
  here we obtain a convex combination of subgraphs.}.  But as we
mentioned earlier, their proof is constructive but might require
exponential time.  The only known result on this problem before this
work is applying Wolsey \cite{wolsey1980heuristic}'s decomposition
which implies that the everywhere $\frac{3}{4}$ vector for a
4-edge-connected 4-regular graph can be decomposed into a convex
combination of 2-edge-connected spanning multigraphs in polynomial
time. However, this is weaker in terms of both the factor and the fact
that we now allow doubled edges. By applying Theorem \ref{cjr} we can
slightly improve this.  The proof of the following theorem is very
similar to the proof of Theorem \ref{2ec8/9cons}.
	\begin{theorem}\label{4ec4reg}
		Let $G = (V,E)$ be a 4-edge-connected, 4-regular
                graph. The everywhere $\frac{3}{4}$ vector for $G$ dominates a convex combination of 2-edge-connected subgraphs of $G$ and this convex combination can be found in polynomial time.
	\end{theorem}
\begin{proof}
	Since $G$ is a 4-edge-connected 4-regular graph the everywhere
        $\frac{1}{2}$ vector for $G$, call it $v$, is in \subt$(G)$. Therefore, $v$ can be written as convex combination of spanning trees of $G$: $v\geq \sum_{i=1}^{\ell} \lambda_i \chi^{T_i}$, where $\lambda_1,\ldots,\lambda_{\ell}$ are convex multipliers and $T_1,\ldots,T_\ell$ are spanning trees of $G$. For $i\in \{1,\ldots,\ell\}$, define a vector $y^i$, where $y^i_e=0$ if $e\in T_i$ and $y^i_e=\frac{1}{3}$ if $e\notin T_i$. Since $G$ is 4-edge-connected we have $y^i\in$ {\sc Cover}$(G,T_i)$. By Theorem \ref{cjr}, we can find 1-covers $C^i_1,\ldots,C^i_{\ell_i}$ of $T_i$ for $i\in\{1,\ldots,\ell\}$ with convex multipliers $\lambda^i_1,\ldots,\lambda^i_{\ell_i}$ such that $\frac{3}{2}y^i = \sum_{j=1}^{\ell_i}\lambda^i_j \chi^{C^i_j}$ for $i\in \{1,\ldots,\ell\}$. Now $T_i+C^i_j$ is a 2-edge-connected subgraph of $G$ for $i\in\{1,\ldots,\ell\}$ and $j\in\{1,\ldots,\ell_i\}$. Let
	$$u=\sum_{i\in \{1,\ldots,\ell\}} \sum_{j\in \{1,\ldots,\ell_i\}}
	\lambda_i \lambda^i_j \chi^{T_i\cup C^i_j}, \quad \text{where
	}\sum_{i\in \{1,\ldots,\ell\}} \sum_{j\in \{1,\ldots,\ell_i\}}
	\lambda_i \lambda^i_j =1.$$
	We can write $u$ as a convex combination of 2-edge-connected subgraphs of $G$. Also,
	
	$$u_e =
	\sum_{i: e\in T_i} \sum_{j=1}^{\ell_i}\lambda_i\lambda^i_j + \sum_{i:
		e\notin T_i} \sum_{j: e\in C^i_j}\lambda_i\lambda^i_j = \frac{1}{2}
	+ \frac{1}{2} \cdot \frac{1}{3}\cdot \frac{3}{2}= \frac{3}{4}.$$
\end{proof}

\section{A Bit Beyond Uniform Covers: Node-Weight Metrics}\label{sec:BIT}

Theorems \ref{18/19} and \ref{15/172ec}, which we proved in Section
\ref{sec:uniform-cover}, imply that when $G$ is a 3-edge-connected,
cubic graph and the everywhere $\frac{2}{3}$ vector is an optimal
solution for \ref{subtour-equal}, we can efficiently find a tour and a
2-edge-connected spanning multigraph whose costs are at most
$\frac{27}{19}$ and $\frac{45}{34}$, respectively, times that of an
optimal solution.  A 3-edge-connected, cubic graph with node-weight
function $f:V\rightarrow \mathbb{R}^+$ falls into this category, as we
will show later on.  However, for such graphs we can obtain
approximation guarantees better than $\frac{27}{19}$ and
$\frac{45}{34}$ for the respective problems.  The techniques we use to
show this are similar to those used in Section
\ref{sec:uniform-cover}. However, these techniques do not generalize
to cubic graphs that have 2-edge cuts. In order to obtain improved
approximation algorithms for this more general class of graphs, we
introduce a connector decomposition theorem.  We use this
decomposition theorem to design algorithms for 2EC and TSP on
node-weighted, subcubic graphs.

\subsection{3-Edge-Connected Cubic Graphs}

First we show that the everywhere $\frac{2}{3}$ vector is in fact an
optimal solution for \ref{subtour} when $G$ is a cubic,
3-edge-connected graph.

\begin{lemma} Let $G=(V,E)$ be a 3-edge-connected,
    cubic graph with node-weight function $f:V \to \mathbb{R}^+$.
    Then $z_G=2\cdot \sum_{v\in V}f_v$.
\end{lemma}
\begin{proof}\label{zg2/3}
	For any $x\in$ \ref{subtour}, we have $x(\delta(v))\geq 2$. So,
	\begin{equation*}
\sum_{e\in E}w(e) x_e ~=~ \sum_{v\in V} x(\delta(v))\cdot f_v ~\geq~ 2 \cdot
\sum_{v\in V}f_v.
\end{equation*}
	Thus, $z_G\geq 2\cdot \sum_{v\in V}f_v$. On the other hand,
        let $x'_e$ denote the everywhere $\frac{2}{3}$ vector for $G$.
 Note that $x'\in$
        \ref{subtour}, since $G$ is 3-edge-connected.  Moreover,
        $\sum_{e\in E}w(e)x'_e = 2\cdot \sum_{v\in V}f_v$.  Hence
        $z_G\leq 2\cdot \sum_{v\in V}f_v$.
\end{proof}
Thus, we see that we can achieve a $\frac{27}{19}$-approximation for
TSP on node-weighted, cubic, 3-edge-connected graphs.  We now show in
fact this approximation ratio can be improved in this special case.
We start with the following observations.

\begin{fact}\label{fact:cc}
Let $C$ be a cycle cover of $G$.  Then
$\sum_{e\in
  C}w(e)= 2\cdot \sum_{v\in V}f_v = z_G.$
\end{fact}

\begin{fact}\label{fact:pm}
Let $M$ be a perfect matching of $G$.  Then
$\sum_{e\in M}w(e) =  \sum_{v\in V}f_v=\frac{z_G}{2}.$
\end{fact}

\introSevenFifths*

\begin{proof}
Let $C$ be a cycle cover of $G$ that covers all 3-edge and 4-edge cuts
of $G$. By Observation \ref{5ec}, the graph $G/C$ is
5-edge-connected. Let $y_e = \frac{2}{5}$ if $e\in E(G/C)$, and
$y_e=0$ otherwise. Notice that $y\in$ \subt$(G/C)$, since for every $S
\subset V(G/C)$, we have $y(\delta(S))\geq \frac{2}{5}\cdot 5\geq 2$.
By Fact~\ref{fact:spanning-tree}, $y$ dominates a convex combination
of spanning trees of $G/C$.  Let $T$ be a minimum spanning tree of
$G/C$.
\begin{align*}
\sum_{e\in T}w(e) &\leq  \sum_{e\in E(G/C)}w(e) y_e & \nonumber\\
& \leq \sum_{e \in E\setminus{C}} w(e) y_e& (E(G/C)\subseteq E\setminus C)\\
& \leq \sum_{e \in E\setminus{C}} w(e) \cdot \frac{2}{5}&(y_e\leq\frac{2}{5}\mbox{ for } e\in E\setminus C) \\
& =  \frac{z_G}{2} \cdot
\frac{2}{5} ~ = ~ \frac{z_G}{5}
& (\mbox{By Fact }\ref{fact:pm}; ~E\setminus C \mbox{ is a perfect matching
  of $G$}).
\end{align*}
Finally, note that $C\cup 2T$ is a tour of $G$ and
\begin{equation*}\sum_{e\in C\cup 2T}w(e)\leq \sum_{e\in C}w(e)+2\cdot\sum_{e\in  T}w(e)\leq z_G+\frac{2}{5}z_G=\frac{7}{5}z_G.\end{equation*}

\end{proof}

Next we show that we can use a very similar approach to 2EC
on node-weighted, 3-edge-connected, cubic graphs.

\begin{theorem}\label{13/10}
There is a  $\frac{13}{10}$-approximation algorithm for 2EC on node-weighted, 3-edge-connected, cubic graphs.
\end{theorem}
\begin{proof}
Let $C$ be a cycle cover of $G$ that covers all 3-edge and 4-edge cuts
of $G$. By Observation \ref{5ec} graph $G/C$ is 5-edge-connected. For
$e\in E(G/C)$ let $y_e = \frac{2}{5}$, and $y_e=0$ otherwise. Notice
that $y\in$ \subt${(G/C)}$. By Christofides' algorithm, one can find a
2-edge-connected multigraph $F$ on $G/C$, such that $\sum_{e\in
  F}w(e)\leq\frac{3}{2} \sum_{e\in E(G/C)}w(e) y_e$. In particular,

\begin{align*}
\sum_{e\in F}w(e) &
\leq \frac{3}{2} \sum_{e\in E(G/C)} w(e) y_e&\\
& \leq \frac{3}{2} \sum_{e\in E\setminus C} w(e) y_e&(E(G/C)\subseteq E\setminus C)\\
& \leq  \frac{3}{5} \sum_{e\in E\setminus C}w(e)&(y_e\leq\frac{2}{5}\mbox{ for } e\in E\setminus C)\\
& =  \frac{3}{10} z_G
& (\mbox{By Fact }\ref{fact:pm}; ~E\setminus C \mbox{ is a perfect matching of $G$}).
\end{align*}

Note that $C\cup F$ is a 2-edge-connected multigraph of $G$ and
\begin{equation*}\sum_{e\in C\cup F}w(e)\leq \sum_{e\in C}w(e)+\sum_{e\in  F}w(e)\leq z_G+\frac{3}{10}z_G=\frac{13}{10}z_G.
\end{equation*}
\end{proof}

We note that for the 2EC problem on 3-edge-connected cubic graphs,
there are better (i.e., smaller) bounds on the integrality gap than
those implied by Theorem \ref{13/10}.  In particular, Boyd and
Legault~\cite{philip6/7} and Legault~\cite{philip} gave bounds of
$\frac{6}{5}$ and $\frac{7}{6}$, respectively, on the integrality gap.
While their procedures are constructive, they do not run in polynomial
time.  Thus, the best previously known approximation factor for this
problem is $\frac{3}{2}$ via Christofides algorithm. Finally one can
easily obtain the following theorem using the ideas in the above
theorems together with Observation \ref{bipartitebit}.

\begin{theorem}\label{4/3bip}
	 	There is a $\frac{4}{3}$-approximation
                (respectively, $\frac{5}{4}$-approximation) algorithm for TSP
                (respectively, 2EC) on node-weighted, 3-edge-connected, cubic,
                bipartite graphs.
\end{theorem}

\subsection{A Tool for Covering 2-Edge Cuts}\label{sec:decomposition}

The results in Theorems \ref{7/5} and \ref{13/10} do not apply to
bridgeless, cubic graphs.  In this section, we give an alternative
tool to the BIT cycle cover (from Theorem \ref{bit13}) for graphs that are not
3-edge-connected (i.e., graphs that contain 2-edge cuts).  In
particular, we find a decomposition 
of a point $x^*$ in
\ref{subtour} such that this decomposition has certain properties.  
Many approaches for TSP
decompose $x^*$ into a convex combination of spanning trees, whose
average weight does not exceed $z_G$.  In this section, we propose an
alternate way of decomposing $x^*$ into {\em connectors}.

\begin{define}
A \textit{connector} $F$ of graph $G$ is a (multi) subset of edges of $G$
such that $F$ is connected and spanning and contains at most two
copies of each edge in $G$.
\end{define}

It is known that a vector $x^* \in$ \ref{subtour} dominates a convex
combination of spanning trees (and hence connectors) of $G$.  We now
show that $x^*$ can be decomposed into connectors with the additional
property that every 2-edge cut is covered an even number of times.
These connectors can be augmented to obtain a tour or a
2-edge-connected multigraph of $G$, and under certain conditions, this
property can be exploited to bound the weight of an augmentation.

\begin{theorem}\label{decomposition}

  Let $x^*\in $ \ref{subtour}. We can find a family of connectors ${\cal{F}}=\{F_1, \ldots
  , F_{\ell}\}$ and multipliers $\lambda_1,\ldots,\lambda_{\ell}$, in
  polynomial-time in the size of the graph $G$, such that
 \begin{itemize}

\item[(a)] $x^* \geq \sum_{i=1}^{\ell} \lambda_i F_i$,   where $\sum
  \lambda_i = 1$ and $\lambda_i > 0$, and

\item[(b)] every $F_i$ has an even number of edges crossing each
  2-edge cut in G.

\end{itemize}
\end{theorem}

We note that $G$ can be assumed to be the support of $x^*$, so every
$F_i$ will actually have an even number of edges crossing each 2-edge
cut in the support of $G$ on $x^*$.

\subsubsection{Proof of Theorem \ref{decomposition}}
	
To prove Theorem \ref{decomposition}, we need to understand the
structure of 2-edge cuts in a 2-edge connected graph.  Assume $G=(V,E)$
is a 2-edge-connected graph. For $S\subseteq V$, let $G[S]$ denote the
subgraph induced by vertex set $S$ (i.e., the graph on the vertex set $S$ containing edges from $E$ with both endpoints in $S$).
\begin{lemma}\label{connectedinside}
If $S\subseteq V$ and $|\delta(S)|=2$, then $G[S]$ is connected.
\end{lemma}
\begin{proof}
Suppose not, then $S$ can be partitioned into $S_1$ and $S_2$, such
that there is no edge in $G$ between $S_1$ and $S_2$. Hence,
$|\delta(S_1)|+|\delta(S_2)|=2$. However, since $G$ is
2-edge-connected we have $|\delta(S_1)|+|\delta(S_2)|\geq 4$, which is
a contradiction.
\end{proof}

\begin{lemma}\label{2edgecutsarenice}
Let $e, f$ and $g$ be distinct edges of $G$. If $\{e,f\}$ and
$\{f,g\}$ are each 2-edge cuts in $G$, then $\{e,g\}$ is also a 2-edge
cut in $G$.
\end{lemma}

\begin{proof}
Let $S,T\subset V$ be such that $\delta(S)=\{e,f\}$ and
$\delta(T)=\{f,g\}$.
Without loss of generality, we can assume that
neither endpoint of $e$ belongs to $T$.  (If both endpoints of $e$
belong to $T$, we set $T$ equal to its complement.)  Moreover, we can
assume that $S \cap T \neq \emptyset$ (since otherwise we can set $S$
equal to its complement).  We can also assume that
$S\setminus T\neq \emptyset$ (since  one endpoint of $e$
belongs to $S$ but not to $T$). Suppose $T\setminus S$
is not empty. By Lemma \ref{connectedinside}, $G[T]$ is
connected. Hence there exists an edge $h$ from $S\cap T$ to
$T\setminus S$. Notice $h\in \delta(S)$, and $h\notin
\delta(T)$. Therefore, $h=e$.  However, since both endpoints of $h$ are in
$T$, this is a contradiction.  So we can assume that $T \setminus S =
\emptyset$.  In other words, $T \subset S$.

Now we show that $\delta(S\setminus T)= \{e,g\}$.  Since $T \subset S$
and neither endpoint of $e$ belongs to $T$, it follows that $e \in
\delta(S \setminus T)$.  Moreover, since only one endpoint of $g$
belongs to $T$ (and therefore to $S$) and $g \notin \delta(S)$, it
follows that $g \in \delta(S \setminus T)$.  So we have
$\{e,g\}\subseteq \delta(S\setminus T)$. Suppose there is another edge
$h\in \delta(S\setminus T)$ with endpoints $v\in S\setminus T$ and
$u\notin S\setminus T$. Note that $h \neq f$, because neither endpoint
of $f$ belongs to $S \setminus T$.
If $u\in T$, then $h\in
\delta(T)$ which is a contradiction to $T$ being a 2-edge
cut. Otherwise if $u\in V\setminus S$, then $h\in \delta(S)$ which is
again a contradiction to $S$ being a 2-edge cut.
\end{proof}

We will later use these properties when building a family of
connectors to delete and replace edges along
the 2-edge cuts of the graph.  Next, we need a decomposition lemma for
$x^*$.

\begin{lemma}\label{decom}
A vector $x^*\in$ \ref{subtour} can be represented as a convex
combination of connectors of $G$, and the number of connectors in
this convex combination is
polynomial in the number of vertices of $G$.
\end{lemma}

\begin{proof}
By Corollary 50.8a in \cite{schrijver} the following polytope is the convex hull of connectors of $G$.
\begin{align*}
x(\delta(\mathcal{P}))& \geq |\mathcal{P}|-1 &\mbox{for }& \mathcal{P}\in \Pi_{n}\\
0 & \leq x_e \leq 2 & \mbox{for }& e\in E
\end{align*}
Here, $\Pi_n$ is the collection of partitions of $V$.
For $\mathcal{P}\in \Pi_n$, we denote by $\delta(\mathcal{P})$ the set
of edges with endpoint in different parts of partition $\mathcal{P}$,
and $|\mathcal{P}|$ is the number of parts in partition
$\mathcal{P}$. Notice that for any partition $\mathcal{P}$ of $V$ with parts $P_1,\ldots,P_{|\mathcal{P}|}$ we have
\begin{equation*}
x^*(\delta(\mathcal{P})) = \frac{1}{2} \sum_{i=1}^{|\mathcal{P}|} x^*(\delta(P_i)) \geq |\mathcal{P}|.
\end{equation*}
Therefore, $x^*$ can be written as a convex combination of connectors
of $G$. The fact that the number of connectors in the convex
combination is polynomial follows from the fact that the polytope above is
separable, and hence we can apply the constructive version of Carath{\'e}odory's theorem
to get the result \cite{lgs,schrijver}.
\end{proof}

By Lemma \ref{decom}, there exists positive reals
$\lambda_1,\ldots,\lambda_{\ell}$, such that $\sum_{i=1}^{\ell}\lambda_i=1$,
and connectors $F_1,\dots,F_{\ell}$ such
that \begin{equation}x^*=\sum_{i=1}^{\ell} \lambda_i
  \chi^{F_i},\end{equation} where $\chi^{F_i}$ is the
characteristic vector of $F_i$ for $i \in \{1,\ldots,\ell\}$. Furthermore, we
can find this decomposition in time polynomial in the size of
$G$. Notice $F_1,\dots,F_{\ell}$ satisfy (a) in the statement of
Theorem \ref{decomposition}.  We will now show that given $F_1, \ldots,
F_{\ell}$, we can obtain a new family of connectors satisfying both (a) and
(b) from Theorem \ref{decomposition}.

\begin{lemma}\label{transform1}
Given a family of connectors $F_1,\ldots,F_{\ell}$ of $G$ such that
$x^*= \sum_{i=1}^{\ell} \lambda_i \chi^{F_i}$, $\lambda_i>0$ for
$i \in \{1,\ldots,\ell\}$, and $\sum_{i=1}^{\ell}\lambda_i=1$, there is a
polynomial-time algorithm that outputs connectors
$F'_1,\ldots,F'_{\ell}$ such that
\begin{itemize}
\item[(1)] $x^*= \sum_{i=1}^{\ell} \lambda_i \chi^{F'_i}$.

\item[(2)] If $x^*_e\geq 1$, then $\chi^{F'_i}(e)\geq 1$ for all $i\in\{1,\ldots,\ell\}$.

\item[(3)] If $x^*_e<1$, then there is no $i\in \{1,\ldots,\ell\}$ such that
  $\chi^{F'_i}(e) = 2$.
\end{itemize}

\end{lemma}
\begin{proof}


Call a tuple $(e,i,j)$ where $e\in E$, $i,j\in\{1,\dots,\ell\}$  {\em bad} if
$$\chi^{F_i}(e)=2 \mbox{ and } \chi^{F_j}(e)=0.$$ Let $m$ be the
number of bad tuples and let $(e,i,j)$ be a bad tuple.
Then
$$F'_i=F_i - e, \; F'_j=F_j + e, \mbox{ and } F'_p = F_p \mbox{ for
} p \in \{1,\dots,\ell\}\setminus \{i,j\}$$ satisfies property (1).  Notice that now $F'_1,\ldots,F'_{\ell}$ has at most $m-1$ bad
tuples; no new bad tuples are created by the above procedure.
Thus, after at most $m$ iterations, we have that
for each $e\in E$, there is no $i,j\in \{1,\dots,\ell\}$ such
  that $\chi^{F'_i}(e)=2$ and $\chi^{F'_j}(e)= 0$.  This implies
  properties (2) and (3) in the statement of the lemma.
Finally, it is also easy to see that fixing each tuple can be done in
polynomial time, and that the number of tuples is polynomial in the
size of $G$.
\end{proof}

We now proceed to the proof of Theorem \ref{decomposition}.  By Lemma
\ref{2edgecutsarenice}, the relation ``is in a 2-edge cut with'' is
transitive. So, we can partition the edges in 2-edge cuts of $G$ into
equivalence classes via this relation. Let $\mathcal{D}$ be the
collection of disjoint subsets of edges of $G$ such that for all $D\in
\mathcal{D}$: (i) $|D|\geq 2$, and (ii) for each pair of edges
$\{e,f\} \subseteq D$, edges $e$ and $f$ form a 2-edge cut of $G$.
Note that for $D\in \mathcal{D}$ and any distinct edges $e, f \in D$,
it cannot be the case that both $x^*_e<1$ and $x^*_f<1$, since
$\{e,f\}$ is a 2-edge cut and $x^*\in$ \ref{subtour}.  We classify the
subsets in $\mathcal{D}$ into two types:
\begin{align*}
\mathcal{D}_1 & = \{D\in\mathcal{D}: \text{ for all } e\in D,
~x^*_e\geq 1 \},\\
\mathcal{D}_2 & = \{D\in\mathcal{D}: \text{ there is exactly one
  edge } e\in D \text{ such that }x^*_e<1 \}.
\end{align*}
Let $F_1, \dots, F_{\ell}$ be a family of connectors satisfying properties
(1), (2) and (3) in Lemma \ref{transform1}.  We propose a procedure to
modify these connectors and output $F'_1, \dots, F'_{\ell}$ such that for
each $D\in \mathcal{D}$, property (b) in Theorem \ref{decomposition}
is satisfied while property (a) is preserved.
In particular, by property (1) from Lemma \ref{transform1}, we have
$$\sum_{i=1}^{\ell} \chi^{F_i}(e)= x^*_e \mbox{ for $e\in E$}.$$
Our specific procedure
depends on whether $D\in \mathcal{D}_1$ or $D\in \mathcal{D}_2$.

\vspace{3mm}

\textbf{Case 1 ($D\in \mathcal{D}_1$):} In this case, we have
$\chi^{F_i}(e)\geq 1$ for all $e\in D$ and $i\in\{1,\dots,\ell\}$, by
property (2) in Lemma \ref{transform1}. For $i \in \{1,\ldots,\ell\}$
let $F'_i$ be such that
$$\chi^{F'_i}(e)=1  \mbox{ for $e\in D$ and } \chi^{F'_i}(e)=\chi^{F_i}(e)  \mbox{ for $e\in E\setminus D$.}$$
Now we reset $F_1, \dots, F_{\ell} := F'_1, \ldots, F'_{\ell}$, and proceed to the
next $D \in \mathcal{D}_1$.

It is easy to see that we can apply this procedure iteratively for
$D\in \mathcal{D}_1$. This is because after applying this
operation on $D\in \mathcal{D}_1$, properties (2) and (3) in Lemma
\ref{transform1} are preserved. Moreover, property (1) in Lemma
\ref{transform1} is also preserved for every edge not in $D$, i.e.

\begin{align*}
\sum_{i=1}^{\ell}\lambda_i \chi^{F'_i}(e) = x^*_e \mbox{ for all $e\in
  E\setminus D$}\quad \quad \mbox{(and $\sum_{i=1}^{\ell}\lambda_i \chi^{F'_i}(e) \leq x^*_e \mbox{ for all $e\in
  	D$}$).}
\end{align*}
 In addition, given any 2-edge cut $\{e,f\}$ such
that $\{e,f\} \subseteq D$ for $D\in \mathcal{D}_1$, we have
$\chi^{F'_i}(e)+\chi^{F'_i}(f)=1+1=2$ for all $i\in
\{1,\dots,\ell\}$.

\vspace{3mm}

\textbf{Case 2 ($D\in \mathcal{D}_2$):} Let $e$ be the unique edge in
$D$ with $x^*_e<1$.  By property (3) in Lemma \ref{transform1}, we
have $\chi^{F_i}(e)\leq 1$ for all $i\in \{1,\ldots,\ell\}$. Without loss
of generality, assume
for $\chi^{F_i}(e)=1$ for $i \in \{1,\ldots,p\}$ and $\chi^{F_i}(e)=0$ for
$i \in \{p+1,\ldots,\ell\}$. For $i\in \{1,\ldots,p\}$, let $F'_i$ be such that
$$\chi^{F'_i}(f)=1 \mbox{ for $f\in D$ and } \chi^{F'_i}(f)=\chi^{F_i}(f)  \mbox{ for $f\in E\setminus D$.}$$
For $i \in \{p+1,\ldots,\ell\}$, let $F'_i$ be such that
$$\chi^{F'_i}(e)=0,\; \chi^{F'_i}(f)=2 \mbox{ for $f\in D\setminus
  \{e\}$} \text{ and } \chi^{F'_i}(f)=\chi^{F_i}(f)  \mbox{ for $f\in
  E\setminus D$.}$$
Now we reset $F_1, \dots, F_{\ell} := F'_1, \ldots, F'_{\ell}$, and proceed to the
next $D \in \mathcal{D}_2$.
After each iteration, we observe that
\begin{align}
\sum_{i=1}^{\ell}\lambda_i \chi^{F'_i}(e)&= \sum_{i=1}^{p}\lambda_i
\chi^{F'_i}(e)+  \sum_{i=p+1}^{\ell}\lambda_i \chi^{F'_i}(e) \nonumber\\
&=\sum_{i=1}^{p}\lambda_i =x^*_e. \label{last-lambda}
\end{align}
For $f\in D\setminus \{e\}$, we have
\begin{align*}
\sum_{i=1}^{\ell}\lambda_i \chi^{F'_i}(f)&= \sum_{i=1}^{p}\lambda_i \chi^{F'_i}(f)+  \sum_{i=p+1}^{\ell}\lambda_i \chi^{F'_i}(f)&\\
&=\sum_{i=1}^{p} \lambda_i +2\sum_{i=p+1}^{\ell}\lambda_i& \\
&= x^*_e+ 2(1-x^*_e)&  \mbox{(From \eqref{last-lambda})}\\
&=2-x^*_e&\\
&\leq x^*_f& \mbox{(Since $x^*\in$ \subt$(G)$).}
\end{align*}
This also clearly holds for any $f\in E\setminus D$ as we do not touch
these edges.  Note that after the final iteration,
$F_1,\ldots,F_{\ell}$ are connected, spanning
multigraphs of $G$, because we began with connected,
spanning multigraphs and we only remove an edge $f$ from $F_i$ if it
contained at least two copies of $f$.

Finally, note that given any 2-edge cut $\{e,f\} \in D$ for $D \in
\mathcal{D}_2$, we have $\chi^{F_i}(e)+ \chi^{F_i}(f)=1+1=2$,
$\chi^{F_i}(e)+ \chi^{F_i}(f)=0+2=2$ or $\chi^{F_i}(e)+
\chi^{F_i}(f)=2+2=4$ for all $i \in \{1,\ldots,\ell\}$.  This concludes
the proof of Theorem \ref{decomposition}.

\subsection{Subcubic Graphs}
We now present two applications of Theorem \ref{decomposition}.  In
the first application, we show that for a node-weighted, subcubic
graph, Christofides' algorithm has an approximation factor better than
$\frac{3}{2}$ when the weight of an optimal subtour solution is
strictly larger than twice the sum of the node weights.  In the second
application, we show that there is a set of edges that can be added to
a connector to yield a 2-edge-connected graph, and this addition can
be found via an application of the tree augmentation problem, which we
introduced in Section \ref{subsubsec:TAP}. This resembles methods used
in the proof of Theorem \ref{2ec8/9cons}.  We then show that combining
the approaches in these applications, we can beat the approximation
ratio of Christofides' algorithm for 2EC on node-weighted, subcubic
graphs.  

A useful fact about node-weighted, subcubic graphs is that the total
edge weight cannot be too much larger than $z_G$. 

\begin{fact}\label{zG-fact2}
Let $G=(V,E)$ be a node-weighted, subcubic graph.  Then $w(E) \leq
\frac{3}{2} z_G$.
\end{fact}

\begin{cproof}
Observe that $w(E) \leq 3 \cdot \sum_{v \in V} f_v$, where $f: V
\rightarrow \mathbb{R}^+$ is the node-weight function.  Also, notice
that $z_G \geq 2 \cdot \sum_{v \in V} f_v$.
\end{cproof}

Since all graphs are assumed to be 2-vertex-connected
(i.e., bridgeless), we can show the following fact.

\begin{fact}\label{zG-fact1}
Let $G=(V,E)$ be a node-weighted, subcubic graph.  Then $z_G \leq 3
\cdot \sum_{v \in V} f_v$.
\end{fact}

\begin{cproof}
This follows
from the fact that $x_e=1$ for all $e \in E$ is a feasible solution
for \subt$(G)$ when $G$ is a 2-vertex-connected subcubic graph.  
\end{cproof}

For the remainder of this section, let $x^*$ be an optimal solution for \ref{subtour}. By Theorem \ref{decomposition}, we have
$x^*\geq\sum_{i=1}^{\ell}\lambda_i\chi^{F_i}$ where $F_i$ is a
connector satisfying (a) and (b) in the statement of Theorem
\ref{decomposition} for $i \in \{1,\ldots,\ell\}$. Let $\xcon  =\sum_{i=1}^{\ell}\lambda_i\chi^{F_i}$. Clearly $\sum_{e \in E}w(e)\xcon_e \leq z_G$. Define $\xcut\in \mathbb{R}^E$ as follows: $\xcut_e=\min\{1,\xcon_e\}$.

\subsubsection{An Algorithm for TSP \`a la Christofides with Simple Deletions}\label{sec:ala-christofides}

In the graph metric, every (minimum) spanning tree has weight at most
$n$.  It follows that in the case where $z_G \geq (1+\epsilon)n$,
Christofides' algorithm has an approximation guarantee strictly better
than $\frac{3}{2}$ (in fact, at most
$(\frac{3}{2}-\frac{\epsilon}{1+\epsilon})$).  This implies that, in
some sense, the most difficult case for graph-TSP is when $z_G = n$.
It seems that it should also be the case for node-weighted graphs: the
most difficult case should be when $z_G = 2\cdot \sum_{v\in V}f_v$,
and when $z_G \geq (1+\epsilon) \cdot 2 \cdot \sum_{v \in V} f_v$,
Christofides' algorithm should give an approximation guarantee
strictly better than $\frac{3}{2}$.

However, in the case of node-weighted graphs (even for subcubic
graphs), a minimum spanning tree of $G$ may have weight exceeding
$2\cdot \sum_{v\in V}f_v$ when $z_G >2\cdot \sum_{v\in V}f_v$.  See
Figure \ref{fig:big-MST} for an example.  Thus, proving an
approximation factor strictly better than $\frac{3}{2}$ for
node-weighted graphs in this scenario does not follow the same
argument as in the graph metric.  Nevertheless, we can use connectors
to prove that we can beat Christofides' algorithm when $G$ is a
subcubic node-weighted graph and $z_G$ is much larger than $2 \cdot
\sum_{v \in V} f_v$.

\begin{figure}[h]
	\begin{subfigure}{.5\textwidth}
		\centering
		\includegraphics[width=0.3\linewidth]{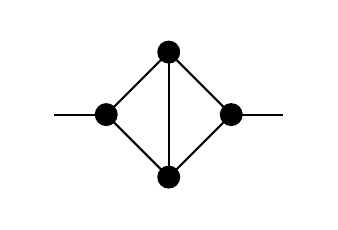}
		\caption{}
		\label{fig:gadget}
	\end{subfigure}
	\begin{subfigure}{.5\textwidth}
		\centering
		\includegraphics[width=0.6\linewidth]{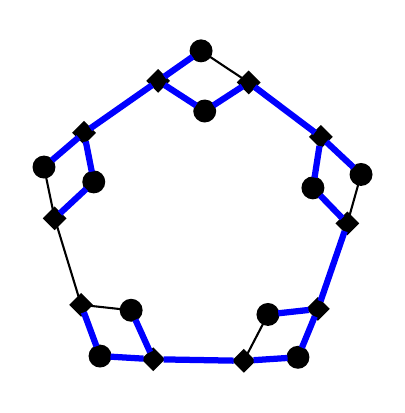}
		\caption{}
		\label{fig:example}
	\end{subfigure}
	\caption{The graph in \subref{fig:example} has a total of $10t$
          (here $t = 5$)
          vertices: each circular vertex corresponds to the gadget
          in \subref{fig:gadget}. The weight of each square vertex
          in \subref{fig:example} is 1,
          and all other vertices have weight zero.  A minimum
          spanning tree (denoted by the thick, blue edges) has weight $5t-2$ while sum of the
          node weights is $2t$. In this case, Theorem \ref{NWcubicTSP} yields a tour of weight $7t-2$, providing a $\frac{7}{5}$-approximation for this instance.}
	\label{fig:big-MST}
\end{figure}

\begin{lemma}\label{epsilon}
Let $G=(V,E)$ be a graph with nonnegative edge weights.  There is an
efficient algorithm to find a tour in $G$ with weight at most $z_G +
\frac{w(E)}{3}$.
\end{lemma}

In fact, we prove something slightly stronger that will be useful
later in the paper.

\begin{lemma}\label{christ-del}
Let $G=(V,E)$ be a graph with nonnegative edge weights.  There is an
efficient algorithm to find a tour in $G$ with weight at most
$\frac{w(E)}{3} + \frac{1}{3} \cdot \sum_{e\in E}w(e) \xcon_e + 
\frac{2}{3} \cdot \sum_{e \in E} w(e)\xcut_e$.
\end{lemma}

For a subset $T$ of vertices in $V$, where $|T|$ is even, a $T$-join
of $G$ is a subgraph $J$ of $G$ in which the set of odd-degree
vertices of $J$ are exactly $T$. Edmonds and Johnson
\cite{edmonds1973matching} proved that the inequalities below
describe the convex hull of $T$-joins of $G$. 
\begin{align}
x(\delta(U)\setminus{W}) - x(W) & \geq 1-|W| &   \mbox{for} & \mbox{ } U
\subseteq V, W \subseteq \delta(U), |U \cap T| + |W| \mbox{ odd}
\tag{{\sc $T$-Join}$(G)$}
\label{t-join-exact}\\
0 \leq x_e & \leq  1    & \mbox{for} & \mbox{ all } e \in E. \nonumber
\end{align}

In Christofides' algorithm, one can write an optimal solution $x^*$
for \ref{subtour} as a convex combination of spanning trees (see Fact
\ref{fact:spanning-tree}).  Each of
these spanning trees is then augmented with a $T$-join, where $T \subseteq
V$ is the set of odd-degree vertices in the spanning tree. In particular,
for a spanning tree $F$ of $G$, let $T$ be the set of odd-degree vertices of
$F$. Then, $\frac{x^*}{2}$ dominates a point in the $T$-join polytope. This mean the vector $x^*+\frac{x^*}{2}=\frac{3}{2}x^*$ dominates a convex combination of tours of $G$.

If we decompose the optimal solution for \ref{subtour} into a family of
connectors according to Theorem \ref{decomposition}, then we can
augment each connector by a $T$-join that is obtained from writing the
vector $\{\frac{1}{3}\}^{E}$ as a convex combination of $T$-joins. 

\begin{lemma}\label{1/3tjoin-exact}
	Let $\mathcal{F}$ be a family of connectors for $G=(V,E)$ satisfying properties (a)
	and (b) from Theorem \ref{decomposition}.  For an $F_i \in
	\mathcal{F}$, let $T$ denote the odd-degree vertices in $F_i$.  Then
	the vector $\{\frac{1}{3}\}^{E}$ belongs to \ref{t-join-exact}.
\end{lemma}

\begin{proof}
	Let $F$ be a connector of $G$ and let $T \subseteq V$ denote the vertices
	with odd degree in $F$.  
Since all edges have value $\frac{1}{3}$, we
	only need to check that
	\begin{align}
	\frac{|\delta(U)|}{3} + \frac{|W|}{3} \geq 1 & \quad \text{ for } U
	\subseteq V, W \subseteq \delta(U), |U \cap T| + |W| \text{ odd}.\label{t-join-main}
	\end{align}
	
	Consider $U \subset V$ such that $|\delta(U)| = 2$.  Note that
        $\sum_{e \in \delta(U)}\chi^F_e$ is even by the properties of
        a connector.  This implies that $|U \cap T|$ is even.  So we
        need to check the case where $|W| = 1$.  In this case, we see
        that Inequality \eqref{t-join-main} is satisfied.  Now
        consider case in which $|\delta(U)| \geq 3$.  In this case,
	$$\frac{|\delta(U)|}{3} + \frac{|W|}{3} \geq \frac{|\delta(U)|}{3} \geq 1.$$
	Hence, $\{\frac{1}{3}\}^E \in$ \ref{t-join-exact}. \end{proof}

Observe that Lemma \ref{1/3tjoin-exact} is sufficient to prove Lemma
\ref{epsilon}.  To prove (the potentially stronger) Lemma
\ref{christ-del}, we modify Christofides' algorithm further by adding
the following {\em deletion} step.  Suppose an edge $e$ occurs in a
connector $F$ as a doubled edge.  If this edge $e$ also belongs to the
$T$-join $J$, we remove two copies of $e$ from the multigraph $F \cup
J$.  We observe that the resulting multigraph remains a tour.

\begin{observation}\label{deleting}
	Let $F$ be a connector for $G=(V,E)$ and let $J$ be a $T$-join, where $T$ is the set of
	vertices with odd degree in $F$.  Let $E' \subset E$ denote the set of
	edges that occur doubled in $F$ and also belong to $J$.  Then the
	multigraph $F \cup
	J \setminus{\{2E'\}}$ is a tour.  
\end{observation}
{}

We are now ready to prove Lemma \ref{christ-del} via an analysis of
the modified Christofides' algorithm we have just described.

\begin{proof}[Proof of Lemma \ref{christ-del}]
	We have $\xcon=\sum_{i=1}^{\ell}\lambda_i\chi^{F_i}$ where $F_i$ is a
        connector satisfying (a) and (b) in the statement of Theorem
        \ref{decomposition} for $i \in \{1,\ldots,\ell\}$. Choose
        $i\in \{1,\ldots,\ell\}$ uniformly at random according to the
        probability distribution defined by
        $\lambda_1,\ldots,\lambda_{\ell}$. Let $T_i$ be the set of
        odd-degree vertices of $F_i$. By Lemma \ref{1/3tjoin-exact},
        we have $\{\frac{1}{3}\}^E =
        \sum_{j=1}^{\ell_i}\lambda^i_j\chi^{J^i_j}$, where ${J^i_j}$
        is a $T_i$-join of $G$. Choose $j\in \{1,\ldots,\ell_i\}$
         at random according to probability distribution
        defined by $\lambda_1^i, \dots, \lambda^i_{\ell_i}$.
 Let $E' \subset E$ denote the edges that occur doubled in $F_i$ and also
belong to $J^i_j$.  By Observation \ref{deleting}, $H = F_i \cup J^i_j
\setminus{\{2E'\}}$ is a tour of $G$. We have 
\begin{align*}
\Ex[w(H)]&=	\Ex[w(F_i)] + \Ex[w(J^i_j)] - 2\cdot\Ex[w(E')]\\
& = \sum_{e \in E}w(e)\xcon_e + \frac{w(E)}{3} - 2 \cdot \sum_{e \in E: \xcon_e > 1} w(e)
\cdot \Pr[\chi_e^{F_i} = 2] \cdot
\Pr[e \in J^i_j]\\
& =  \sum_{e \in E}w(e)\xcon_e+ \frac{w(E)}{3} - 2 \cdot \sum_{e \in E: \xcon_e > 1} w(e) (\xcon_e -1)
\cdot \frac{1}{3}\\
& =  \sum_{e \in E}w(e)\xcon_e+ \frac{w(E)}{3} - \frac{2}{3} \left(\sum_{e \in E: \xcon_e > 1}
w(e) \xcon_e -\sum_{e \in E: \xcon_e > 1}
w(e)\right)\\
& =  \sum_{e \in E}w(e)\xcon_e+ \frac{w(E)}{3} - \frac{2}{3} \left(\sum_{e \in E}
w(e) \xcon_e -\sum_{e \in E}
w(e)\xcut_e\right)\\
& = \frac{\sum_{e\in E}w(e)\xcon_e}{3} + \frac{w(E)}{3} + \frac{2}{3} \cdot \sum_{e \in E} w(e)\xcut_e.
\end{align*}

\end{proof}

\begin{theorem}\label{NWcubicTSP}
	Let $G$ be a node-weighted, subcubic graph. If
	$z_G \geq 2\cdot(1+\epsilon)\cdot \sum_{v\in V}f_v$, then there is an
	$(\frac{3}{2}-\frac{\epsilon}{3})$-approximation algorithm for TSP on $G$.
\end{theorem}
\begin{proof}
For a node-weighted, subcubic graph, we have
	\begin{eqnarray}
	w(E) & \leq & 3 \cdot \sum_{v \in V} f_v. \label{thm11:bound}
	\end{eqnarray}
	By the assumption of the theorem and \eqref{thm11:bound}, we have $z_G\geq
	2(1+\epsilon)\sum_{v \in V} f_v \geq \frac{2(1+\epsilon)}{3}w(E)$.
	Applying Lemma \ref{epsilon}, we get a tour of weight at most
	\begin{eqnarray*}
		z_G+\frac{w(E)}{3} & \leq & (\frac{3+2\epsilon}{2 + 2\epsilon})\cdot z_G\\
		& = & (\frac{3}{2} - \frac{\epsilon}{2 + 2\epsilon})\cdot z_G\\
		& \leq & (\frac{3}{2} -\frac{\epsilon}{3})\cdot z_G.
	\end{eqnarray*}
	The last inequality comes from the fact that $\epsilon \leq
	\frac{1}{2}$ since $z_G\leq 3\cdot \sum_{v\in V}f_v$, which follows
	from Fact \ref{zG-fact1}. \end{proof}
\subsubsection{An Algorithm for 2EC}\label{subsec:tap}

Recall the set-up for 2EC. We are given a graph $G=(V,E)$ with
nonnegative weights $w(e)$ for $e\in E$. Our goal is to find a minimum
weight 2-edge-connected multigraph of $G$.  
We now prove the following lemma. 

\begin{lemma}\label{2ecsubcubic}
Let $G=(V,E)$ be a graph with nonnegative edge weights.
We can find a 2-edge-connected multigraph of $G$ with weight at most
$\sum_{e \in E} w(e)\xcon_e + \frac{2}{3}w(E) - \frac{2}{3}\cdot \sum_{e\in E}w(e)\xcut_e$.
\end{lemma}

\begin{proof}
 Recall that we have $\xcon=\sum_{i=1}^{\ell}\lambda_i\chi^{F_i}$ where $F_i$ is a
 connector satisfying (a) and (b) in the statement of Theorem
 \ref{decomposition} for $i \in \{1,\ldots,\ell\}$.  For $i \in
\{1,\ldots,\ell\}$, let $\mathcal{S}_i$ be the family of 1-edge cuts
of $F_i$. As discussed in Section \ref{subsubsec:TAP}, there is a laminar family $\mathcal{S}^*_i\subseteq \mathcal{S}_i$  that is enough to describe \textsc{Cover}$(G,F_i)$ for all $i\in \{1,\ldots,\ell\}$.  Define vector $y^i\in \mathbb{R}^E$ as follows: $y^i_e=0$ for
$e\in F_i$ and $y^i_e=\frac{1}{2}$ for $e\in E\setminus F_i$.

\begin{claim}\label{claim1}
For $i \in \{i, \dots, \ell\}$, we have $y^i \in$ {\sc
  Cover}$(G,F_i)$.
\end{claim}

\begin{cproof}
Let $S$ be a 1-edge cut of $F_i$.  Then $\delta(S)\cap F_i$ contains
exactly one edge $e$. Note that it cannot be the case that
$|\delta(S)|=2$. This is because if $\delta(S)$ were a 2-edge cut of
$G$, then by property (b) in Theorem \ref{decomposition}, there would
be an even number of edges in $F_i$ that are also in
$\delta(S)$. Hence, $|\delta(S)|\geq 3$. So we have
\begin{equation*}
\sum_{e \in \delta(S)}y_e ~ = \sum_{e \in \delta(S)\setminus  F_i}
\frac{1}{2} ~ = \sum_{e \in \delta(S)\setminus  \{e\}} \frac{1}{2} ~ =
~\frac{|\delta(S)\setminus \{e\}|}{2} ~ \geq ~1.
\end{equation*}
This concludes the proof of the claim.\end{cproof}
For $i\in \{1,\ldots,\ell\}$, define vector $r^i$ as follows: $r^i_e=0$
for $e\in F_i$ and $r^i_e=\frac{2}{3}$ for $e\in E\setminus F_i$.
\begin{claim}\label{claim2}
	For $i \in \{1,\ldots,\ell\}$, the vector $r^i$ can be written as
        a convex combination of 1-covers of $F_i$.
\end{claim}	
\begin{cproof} By Claim \ref{claim1} and Theorem \ref{cjr}, vector $\frac{4}{3}y^i$ can be written as
	a convex combination of 1-covers of $F_i$, and $\frac{4}{3}y^i=r^i$.
\end{cproof}

By Claim \ref{claim2}, for $i \in \{1,\ldots,\ell\}$ we can write
$r^i$ as $\sum_{j=1}^{\ell_i}\lambda^i_j C^i_j$, where for $j \in
\{1,\ldots,\ell_i\}$, $C^i_j$ is a 1-cover for $F_i$.  Let $R^i_j =
F_i \cup C^i_j$.  Notice, for all choices of $i$ and $j$, $R_i^j$ is a
2-edge-connected multigraph of $G$.  To argue that there exists a
low-weight, 2-edge-connected multigraph, we show the following claim.

\begin{claim}\label{claim3}
There exists $i \in \{1,\ldots,\ell\}$ and $j\in \{1,\ldots,\ell_i\}$
such that $R^i_j\leq \sum_{e \in E} w(e)\xcon_e + \frac{2}{3}w(E) -
\frac{2}{3}\cdot \sum_{e\in E}w(e)\xcut_e$.
\end{claim}
\begin{cproof}
Pick $i\in \{1,\ldots,\ell\}$ at random according to the probability
distribution defined by $\lambda_1,\ldots,\lambda_{\ell}$. Now, pick
$j\in \{1,\ldots,\ell_i\}$ at random according to the probability
distribution defined by $\lambda^i_1,\ldots,\lambda^i_{\ell_i}$.
We have
\begin{align*}
\Ex[w(R^i_j)]&= \Ex[w(F_i)] + \Ex[w(C^i_j)]\\
& = \sum_{e \in E}\big( 2 w(e)\cdot\Pr [\chi^{F_i}(e)=2]+ w(e)\cdot\Pr [\chi^{F_i}(e)=1]\big)+\sum_{e\in E}w(e)\cdot \Pr[e\in C^i_j]\\
& = \sum_{e \in E}\big( 2 w(e)\cdot\Pr [\chi^{F_i}(e)=2]+
w(e)\cdot\Pr [\chi^{F_i}(e)=1]\big)+\sum_{e\in E} \frac{2}{3} w(e) \cdot\Pr[\chi^{F_i}(e)=0]\\
& = \sum_{e \in E: \xcon_e>1}\big( 2 w(e)\cdot\underbrace{\Pr [\chi^{F_i}(e)=2]}_{=(\xcon_e-1)}+ w(e)\cdot\underbrace{\Pr [\chi^{F_i}(e)=1]}_{=(2-\xcon_e)}+ \frac{2}{3} w(e) \cdot\underbrace{\Pr[\chi^{F_i}(e)=0]}_{=0}\big)\\
& + \sum_{e \in E: \xcon_e\leq1}\big( 2 w(e)\cdot\underbrace{\Pr [\chi^{F_i}(e)=2]}_{=0}+ w(e)\cdot\underbrace{\Pr [\chi^{F_i}(e)=1]}_{=\xcon_e}+ \frac{2}{3} w(e) \cdot\underbrace{\Pr[\chi^{F_i}(e)=0]}_{=(1-\xcon_e)}\big)\\
&=\sum_{e\in E: \xcon_e>1}\big( 2 w(e)\xcon_e - 2 w(e) +
2 w(e) -w(e) \xcon_e \big)+ \sum_{e\in E: \xcon_e\leq1} \big(  w(e)\xcon_e +\frac{2}{3} w(e)- \frac{2}{3} w(e)\xcon_e \big)\\
&=\sum_{e\in E: \xcon_e>1} w(e)\xcon_e + \sum_{e\in E: \xcon_e\leq1}
\big(\frac{1}{3} w(e)\xcon_e +\frac{2}{3} w(e)\big)\\
& = 
\sum_{e\in E: \xcon_e>1}w(e)(\xcon_e - 1) +
  \sum_{e\in E} (\frac{1}{3}w(e)\xcut_e  +\frac{2}{3}
  w(e))\\
  & = \sum_{e \in E} w(e)\xcon_e - \sum_{e\in E}w(e)\xcut_e+  \sum_{e\in E} \frac{1}{3}w(e)\xcut_e  +\sum_{e\in E}\frac{2}{3}  w(e)\\
  &= \sum_{e \in E} w(e)\xcon_e + \frac{2}{3}w(E) - \frac{2}{3}\cdot \sum_{e\in E}w(e)\xcut_e.
\end{align*}
\end{cproof}

This concludes the proof of Lemma \ref{2ecsubcubic}.
\end{proof}

Assume $w(E)\leq \frac{3}{2} z_G$. In this case, Lemma
\ref{2ecsubcubic} finds a 2-edge-connected multigraph of weight at
most $2z_G- \frac{2}{3} \cdot \sum_{e\in E}w(e)\xcut_e$. If $\sum_{e\in
  E}w(e)\xcut_e=z_G$, then this implies a $\frac{4}{3}$-approximation
for 2EC. (Note that this is the case if $x^*\leq 1$.) However, there
are instances for which this does not happen. Figure \ref{4/3not}
illustrates an example where the algorithm in Lemma \ref{2ecsubcubic}
does not improve the bound of Christofides' algorithm.
\begin{center}
	\begin{figure}[h]
		\centering
\begin{tikzpicture}
\node[shape=circle,draw=black,label={[xshift=-0.6cm, yshift=-0.67cm]$2$}] (A) at (0,0) {$v_1$};
\node[shape=circle,draw=black,label={[xshift=-0.6cm, yshift=-0.67cm]$1$}] (B) at (0,4) {$v_2$};
\node[shape=circle,draw=black,label={[xshift=0.6cm, yshift=-0.67cm]$0$}] (C) at (4,4) {$v_3$};
\node[shape=circle,draw=black,label={[xshift=0.6cm, yshift=-0.67cm]$1$}] (D) at (4,0) {$v_4$};
\path [-] (A) edge node[left,label={[xshift=-0.2cm, yshift=-0.3cm]$\epsilon$}] {} (B);
\path [-] (C) edge node[left,label={[xshift=0.75cm, yshift=-0.4cm]$1-\frac{\epsilon}{2}$}] {} (D);
\path [-] (B) edge node[left,label={[xshift=-0cm, yshift=-0.1cm]$1-\frac{\epsilon}{2}$}] {} (C);
\path [-] (A) edge node[left,label={[xshift=-0cm, yshift=-0.7	cm]$\epsilon$}] {} (D);
\path [-] (A) edge node[left,label={[xshift=-0.28cm, yshift=-1.55	cm]$2-2\epsilon$}] {} (C);
\path [-] (B) edge node[left,label={[xshift=-0.26cm, yshift=0.55	cm]$1-\frac{\epsilon}{2}$}] {} (D);
\end{tikzpicture}
\caption{Let $G=(V,E)$ be the node-weighted $K_4$ shown above. For
  $e\in E$, $w_e$ is defined as the sum of the node-weights of the two
  endpoints (e.g., $w_{v_1v_2}= 2 + 1= 3$). The
  edge labels represents solution $x^*\in $ \ref{subtour}. Here we have $\xcon=x^*$. We have $w(E)= 12$, $\sum_{e \in E}w(e)\xcon_e = 8$, $\sum_{e \in E}w(e)\xcut_e = 6+4\epsilon$. For this $x^*$,
  Lemma \ref{2ecsubcubic} yields a
  $(\frac{3-\epsilon}{2})$-approximation, which does not outperform
  Christofides' algorithm by any constant factor. 
 However, Lemma
  \ref{christ-del} provides a $(\frac{4+\epsilon}{3})$-approximation
  for 2EC on the graph $G$.}
\label{4/3not}
\end{figure}
\end{center}
\begin{lemma}\label{2ecgeneral}
	Let $G=(V,E)$ be a graph such that $w(E)\leq \beta \cdot z_G$, then there is a $(\frac{2}{3}+\frac{\beta}{2})$-approximation for 2EC on graph $G$.
\end{lemma}
\begin{proof}
	Taking the best of the guarantees from Lemmas \ref{christ-del} and
	\ref{2ecsubcubic},we have an algorithm that outputs a
	2-edge-connected multigraph of weight at most 
	\begin{align*}
	&\frac{1}{2}\left(\frac{4}{3}\sum_{e \in E}w(e)\xcon_e +  w(E))
	\right)  \leq \frac{1}{2} \left(\frac{4}{3}z_G + w(E) \right) = (\frac{2}{3}+\frac{\beta}{2})\cdot z_G.
	\end{align*}
	Note that the above bound is obtained by taking the average of the two guarantees.
	\end{proof}

\FourThirdsEC*

\begin{proof}
For a node-weighted, subcubic graph, we have $w(E) \leq
\frac{3}{2}z_G$ (by Fact \ref{zG-fact2}). By Lemma \ref{2ecgeneral}, we get a $\frac{17}{12}$-approximation for 2EC on graph $G$.
\end{proof}

\section{Concluding Remarks}

Carr and Ravi~\cite{Carr98} proved that for any 4-regular
4-edge-connected graph $G$, the everywhere $\frac{2}{3}$ vector can be
decomposed into a convex combination of 2-edge-connected subgraphs of
$G$. This implies an upper bound of $\frac{4}{3}$ on the integrality
gap of half-integer points for the 2EC problem with metric weights.
Their proof however does not lead to a polynomial-time algorithm for
such instances.  In Theorem \ref{4ec4reg}, we gave an alternate way
(as opposed to that of Wolsey \cite{wolsey1980heuristic}) to obtain
such a convex combination for the everywhere $\frac{3}{4}$ vector.  
It is an interesting open problem to determine if
the everywhere $\frac{3}{4}-\epsilon$
vector for $G$ can be decomposed into convex combination tours of $G$
in polynomial time.  Another open problem is stated in Conjecture
\ref{8/9conj-int}, which is implied by the four-thirds conjecture.
Finally, for node-weighted metrics, it would be interesting to find a
$\frac{4}{3}$-approximation algorithm for TSP in bridgeless,
cubic graphs to match the corresponding bound for graph
metrics~\cite{BSSS11}.

\section{Acknowledgements}
We would like to thank Jennifer Iglesias for discussions on the tree
augmentation problem, G\'erard Cornu\'ejols for his comments on a
preliminary draft of this paper, and Sylvia Boyd for clarifications
regarding recent work on the 2EC problem.  We thank an anonymous
referee for pointing out an error in a previous version: we had
claimed a smaller approximation ratio in the statement of
Theorem \ref{thm:2EC-node-weight}.

The work of A. Haddadan and R. Ravi is supported in part by the
U. S. Office of Naval Research under award number N00014-18-1-2099,
and the U. S. National Science Foundation under award number
CCF-1527032.  The work of A. Newman is supported in part by LabEx
PERSYVAL-Lab (ANR 11-LABX-0025) and IDEX-IRS SACRE.  Our joint work
was also supported by a research grant from the Carnegie Bosch
Institute.

\bibliographystyle{unsrt}
\bibliography{decomposition}

\end{document}